\documentclass{amsart}
\pdfoutput=1
\usepackage{amsmath,amsbsy,amssymb,amsthm}
\usepackage{url}
\usepackage{color}
\usepackage{natbib}
\usepackage{graphicx}
\usepackage{thm-restate}
\usepackage{algorithm}
\usepackage{algpseudocode}
\usepackage[hidelinks]{hyperref}
\usepackage{hyphenat}
\algnewcommand\algorithmicinput{\textbf{Input:}}
\algnewcommand\algorithmicoutput{\textbf{Output:}}
\algnewcommand\Input{\item[\algorithmicinput]}
\algnewcommand\Output{\item[\algorithmicoutput]}

\usepackage{booktabs,tabularx}
\usepackage{multirow}

\usepackage{setspace}
\newtheorem{Theorem}{Theorem}[section]
\newtheorem{Definition}[Theorem]{Definition}
\newtheorem{Proposition}[Theorem]{Proposition}

\newtheorem{Assumption}[Theorem]{Assumption}

\newtheorem{Remark}[Theorem]{Remark}
\newtheorem{Lemma}[Theorem]{Lemma}

\numberwithin{equation}{section}

\newcommand{\h}{\hspace*{.24in}}
\newcommand{\mb}{\mathbf}
\newcommand{\mbb}{\mathbb}

\newcommand{\argmin}{\operatornamewithlimits{argmin}}

\newcommand{\ra}[1]{\renewcommand{\arraystretch}{#1}}

\newcommand{\beginsupplement}{%
        \setcounter{table}{0}
        \renewcommand{\thetable}{S\arabic{table}}%
        \setcounter{figure}{0}
        \renewcommand{\thefigure}{S\arabic{figure}}}

\title{\textnormal{Non-bifurcating phylogenetic tree inference via the adaptive LASSO}}

\author[Zhang, Dinh, and Matsen]{Cheng Zhang${}^{1 \ast}$,Vu Dinh${}^{2 \ast}$, and Frederick A. Matsen IV${}^{3}$}
\thanks{
\noindent ${}^{\ast}$ Equal contribution, \\
${}^{1}$ School of Mathematical Sciences and Center for Statistical Science, Peking University,
${}^{2}$ Department of Mathematical Sciences, University of Delaware,
${}^{3}$ Computational Biology Program, Fred Hutchinson Cancer Research Center}
\begin{document}

\begin{abstract}
Phylogenetic tree inference using deep DNA sequencing is reshaping our understanding of rapidly evolving systems, such as the within-host battle between viruses and the immune system.
Densely sampled phylogenetic trees can contain special features, including \emph{sampled ancestors} in which we sequence a genotype along with its direct descendants, and \emph{polytomies} in which multiple descendants arise simultaneously.
These features are apparent after identifying zero-length branches in the tree.
However, current maximum-likelihood based approaches are not capable of revealing such zero-length branches.
In this paper, we find these zero-length branches by introducing adaptive-LASSO-type regularization estimators for the branch lengths of phylogenetic trees, deriving their properties, and showing regularization to be a practically useful approach for phylogenetics.
\end{abstract}

\maketitle

\hspace{0.2cm}


\hspace{0.2cm}

\section{Introduction}
Phylogenetic methods, originally developed to infer evolutionary relationships among species separated by millions of years, are now widely used in biomedicine to investigate very short-time-scale evolutionary history.
For example, mutations in viral genomes can inform us about patterns of infection and evolutionary dynamics as they evolve in their hosts on a time-scale of years \citep{Grenfell2004-dz}.
Antibody-making B cells diversify in just a few weeks, with a mutation rate around a million times higher than the typical mutation rate for cell division \citep{Kleinstein2003-mu}.
Although general-purpose phylogenetic methods have proven useful in these biomedical settings, the basic assumption that evolutionary trees follow a bifurcating pattern need not hold.
Our goal is to develop a penalized maximum-likelihood approach to infer non-bifurcating trees (Figure~\ref{fig:motivation}).

Although our practical interests concern inference for finite-length sequence data, some situations in biology will lead to non-bifurcating phylogenetic trees, even in the theoretical limit of infinite sequence information.
For example, a retrovirus such as HIV incorporates a copy of its genetic material into the host cell upon infection.
This genetic material is then used for many copies of the virus, and when more than two descendants from this infected cell are then sampled for sequencing, the correct phylogenetic tree forms a multifurcation from these multiple descendants (a.k.a.\ a \emph{polytomy}).
In other situations we may sample an ancestor along with a descendant cell, which will appear as a node with a single descendant edge (Figure~\ref{fig:motivation}).
For example, antibody-making B cells evolve within host in dense accretions of cells called germinal centers in order to better bind foreign molecules \citep{Victora2012-lu}.
In such settings it is possible to sample a cell along with its direct descendant.
Indeed, upon DNA replication in cell division, one cell inherits the original DNA of the coding strand, while the other inherits a copy which may contain a mutation from the original.
If we sequence both of these cells, the first cell is the genetic ancestor of the second cell for this coding region.
In this case the correct configuration of the two genotypes is that the first cell is a \emph{sampled ancestor} of the second cell.

\begin{figure}
    \centering
    \includegraphics[width=.65\textwidth]{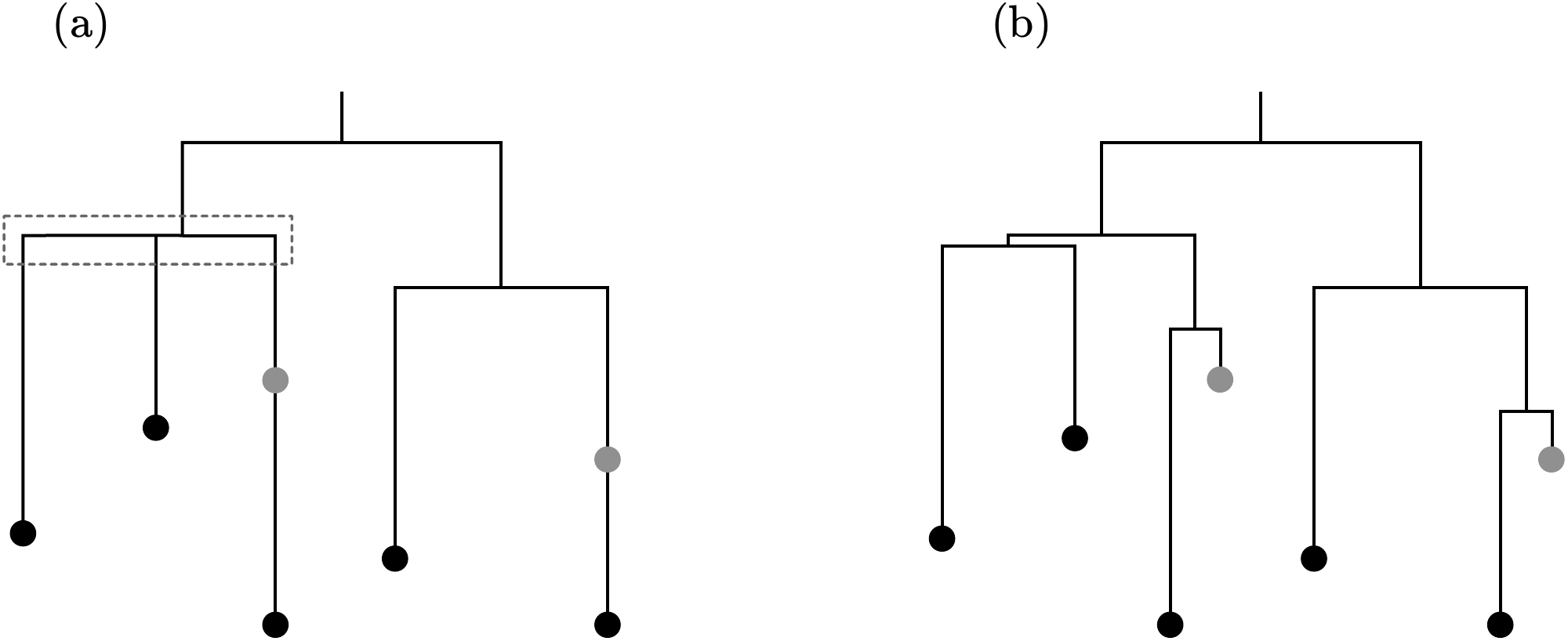}
    \caption{
    (a) A cartoon evolutionary scenario, with sampled ancestors (gray dots) and a multifurcation (dashed box).
    (b) A corresponding standard maximum likelihood phylogenetic inference, without regularization or thresholding.
    }
\label{fig:motivation}
\end{figure}

However DNA sequences are finite and often rather short, limiting the amount of information available with which to infer phylogenetic trees.
Even though entire genomes are large, the segment of interest for a phylogenetic analysis is frequently small.
For example, B cells evolve rapidly only in the hundreds of DNA sites used to encode antibodies, and thus sequencing is typically applied only to this region \citep{Georgiou2014-sh}.
Similarly, modern applications of pathogen outbreak analysis using sequencing \citep{Gardy2015-rb} frequently observe the same sequence, indicating that sampling is dense relative to mutation rates.
Because genetic recombination and processes such as viral reassortment \citep{Chen2008-fh} break the assumption that genetic data has evolved according to a single tree, practitioners often restrict analysis to an even shorter region that they believe has evolved according to a single process.

Inference on these shorter sequences further motivates correct inferences for non-bifurcating tree inference.
Indeed, even if a collection of sequences in fact did diverge in a bifurcating fashion, if no mutations happened in the sequenced region during this diversification (i.e.\ a zero-length branch) then a non-bifurcating representation is appropriate.
We thus expect multifurcations and sampled ancestors whenever the interval between the bifurcations is short compared to the total mutation rate in the sequenced region.

Non-bifurcating tree inference has thus far been via Bayesian phylogenetics, with the two deviations from bifurcation in two separate lines of work.
For multifurcations, \citet{Lewis2005-ez,Lewis2015-kv} develop a prior on phylogenetic trees with positive mass on multifurcating trees, and then perform tree estimation using reversible jump MCMC (rjMCMC) moves between trees.
For sampled ancestors, \citet{Gavryushkina2014-qp,Gavryushkina2016-ew} introduce a prior on trees with sampled ancestors and then also use rjMCMC for inference.
To our knowledge no priors have been defined that place mass on trees with multifurcations and/or sampled ancestors.

Current biomedical applications require a more computationally efficient alternative than these Bayesian techniques.
Indeed, current methods for real-time phylogenetics in the course of a viral outbreak use maximum likelihood \citep{Neher2015-jr,Libin2017-xn}, which is orders of magnitude faster than Bayesian analyses.
This is essential because the time between new sequences being added to the database can be shorter than the required execution time for a Bayesian analysis.
However, to our knowledge an appropriate maximum-likelihood alternative to such rjMCMC phylogenetic inference for multifurcating trees does not yet exist.

Elsewhere in statistics, researchers find the set of parameters with zero values via penalized maximum likelihood inference, commonly maximizing the sum of a penalty term and a log likelihood function.
When the penalty term has a nonzero slope as each variable approaches zero, the penalty will have the effect of ``shrinking'' that variable to zero when there is not substantial evidence from the likelihood function that it should be nonzero.
There is now a substantial literature on such estimators, of which $L_1$ penalized estimators such as LASSO \citep{Tibshirani1996-xs} are the most popular.

In this paper, we introduce such regularization estimators into phylogenetics, derive their properties, and show this regularization to be a practically-useful approach for phylogenetics via new algorithms and experiments.
Specifically, we first show consistency: that the LASSO and its adaptive variants find all zero-length branches in the limit of long sequences with an appropriate penalty weight.
We also derive new algorithms for phylogenetic LASSO and show them to be effective via simulation experiments and application to a Dengue virus data set.
Throughout this paper, we assume that the topology of the tree is known, and that the branch lengths of the trees are bounded above by some constant.

Phylogenetic LASSO is challenging and requires additional new techniques above those for classical LASSO.
First, the phylogenetic log-likelihood function is non-linear and non-convex.
More importantly, unlike the standard settings for model selection where the variables can receive both positive and negative values, the branch lengths of a tree are non-negative.
Thus, the objective function of phylogenetic LASSO can only be defined on a constrained compact space, for which the ``true parameter'' lies on the boundary of the domain.
Furthermore the behavior of the phylogenetic log-likelihood on this boundary is untamed: when multiple branch lengths of a tree approach zero at the same time, the log-likelihood function may diverge to infinity, even if it is analytic in the inside of the domain of definition.
The geometry of the subset of the boundary where these singularities happen is non-trivial, especially in the presence of randomness in data.
All of these issues combine to make theoretical analyses and practical implementation of these estimators an interesting challenge.

% We emphasize that a tree with zero length branches give a qualitatively different evolutionary statement about tree branching structure than a tree without.
% A zero length branch ending at observed sequence $x$ and attached to a path in the tree leading to sequence $y$ states that $y$'s ancestor had genotype $x$ in the sequenced region.
% A zero length internal branch states that there was simultaneous divergence of multiple organisms from a given genotype, in other words a multifurcation.
% These are different statements, with different evolutionary interpretations, than the corresponding statements with nonzero branches.

\section{Mathematical framework}
\label{sec:math}

\subsection{Phylogenetic tree}

A \emph{phylogenetic tree} is a tree graph $\tau$ such that each leaf has a unique name, and such that each edge $e$ of the tree is associated with a non-negative number $q_e$.
We will denote by $E$ and $V$ the set of edges and vertices of the tree, respectively.
We will refer to $\tau$ and $(q_e)_{e \in E}$ as the \emph{tree topology} and the \emph{vector of branch lengths}, respectively.
Any edge adjacent to a leaf is called a pendant edge, and any other edge is called an internal edge.
A pendant edge with zero branch length leads to a \emph{sampled ancestor} while an internal edge with zero branch length is part of a \emph{polytomy}.

As mentioned above, we assume that the topology $\tau$ of the tree is known and we are interested in reconstructing the vector of branch lengths.
Since the tree topology is fixed, the tree is completely represented by the vector of branch lengths $q$.
We will consider the set $\mathcal{T}$ of all phylogenetic trees with topology $\tau$ and branch lengths bounded from above by some $g_0 > 0$.
This arbitrary upper bound on branch lengths is for mathematical convenience and does not represent a real constraint for the short-term evolutionary setting of interest here.

%%%%%%%%%%%%%%%%%

\subsection{Phylogenetic likelihood}

We now summarize the likelihood-based formulation of phylogenetics.
The input data for this formulation is a collection of molecular sequences (such as DNA sequences) that have been aligned into a collection of \emph{sites}.
We assume that the differences in sequences between sites is due to a point mutation process that is modeled with a continuous-time Markov chain.
One can use a dynamic program to calculate a likelihood (detailed below), allowing one to select the maximum-likelihood phylogenetic tree and model parameters.

We will follow the most common setting for likelihood-based phylogenetics: a reversible continuous-time Markov chain model of substitution which is IID across sites.
Briefly, let $\Omega$ denote the set of states and let $r= |\Omega|$; for convenience, we assume the states have indices $1$ to $r$.
We assume that mutation events occur according to a continuous-time Markov chain on states $\Omega$.
Specifically, the probability of ending in state $y$ after ``evolutionary time'' $t$ given that the site started in state $x$ is given by the $xy$-th entry of $P(t)$, where $P(t)$ is the matrix valued function $P=e^{Qt}$, and the matrix $Q$ is the instantaneous rate matrix of the evolutionary model.
Here $Q$ is normalized to have mean rate 1, so ``evolutionary time'' $t$ is measured in terms of the expected number of substitutions per site.
We assume that the rate matrix $Q$ is reversible with respect to a stationary distribution $\pi$ on the set of states $\Omega$.

We will use the term \emph{state assignment} to refer to a single-site labelling of the leaf of tree by characters in $\Omega$.
For a fixed vector of branch lengths $q$, the phylogenetic likelihood is defined as follows and will be denoted by $L(q)$.
Let $\mb{Y}^k= (\mb{Y}^{(1)}, \mb{Y}^{(2)},...,\mb{Y}^{(k)}) \in \Omega^{N \times k}$ be the observed sequences (with characters in $\Omega$) of length $k$ over $N$ leaves (i.e., each of the $\mb{Y}^{(i)}$'s is a state assignment).
We will say that a function $f$ \emph{extends} a function $g$ if $f$ has a larger domain than $g$ but agrees with $g$ on its domain.
The likelihood of observing $\mb{Y}$ given the tree has the form
\[
L_{k}(\mb{Y}; q) = \prod_{i=1}^k{\sum_{a^i}{\eta(a_{\rho}^i)\prod_{(u,v)\in E}{P_{a^i_ua^i_v}( q_{uv})}}}
\]
where $\rho$ is any internal node of the tree, $a^i$ ranges over all extensions of $\mb{Y}^{(i)}$ to the internal nodes of the tree, $a^i_u$ denotes the assigned state of node $u$ by $a^i$, $P_{xy}(t)$ denotes the transition probability from character $x$ to character $y$ across an edge of length $t$ defined by a given evolutionary model and $\eta$ is the stationary distribution of this evolutionary model.
The value of the likelihood does not depend on choice of $\rho$ due to the reversibility assumption.

We will also denote $\ell_k(q) = \log(L_k(\mb{Y}; q))$ and refer to it as the log-likelihood function given the observed sequence data.
We allow the likelihood of a tree given data to be zero, and thus $\ell_k$ is defined on $\mathcal{T}$ with values in the extended real line $[-\infty, 0]$.
We note that $\ell_k$ is continuous, that is, for any vector of branch lengths $q_0 \in \mathcal{T}$, we have
\[
\lim_{q \to q_0}{\ell_k(q)} = \ell_k(q_0)
\]
even if $\ell_k(q_0) = - \infty$.

Each vector of branch lengths $q$ generates a distribution on the state assignment of the leaves, hereafter denoted by $P_q$.
We will make the following assumptions:

\begin{Assumption}[Model identifiability] $\h P_q = P_{q'} \h   \Leftrightarrow \h q=q'$.
\label{assump:iden}
\end{Assumption}

\begin{Assumption}
The data $\mb{Y}^k$ are generated on a tree topology $\tau$ with vector of branch lengths $q^* \in \mathcal{T}$ according to the above Markov process, where some components of $q^*$ might be zero.
We assume further that the tree distance (the sum of branch lengths) between any pair of leaves of the true tree is strictly positive.
\label{assump:leafedges}
\end{Assumption}
We note that model identifiability (Assumption \ref{assump:iden}) is a standard assumption and is essential for inferring evolutionary histories from data in the likelihood-based framework.
This condition holds for a wide range of evolutionary models that are used in phylogenetic inference, including the Jukes-Cantor, Kimura, and other time-reversible models \citep{chang1996full,Allman2008-wd, Allman2008-rd}.
The second criterion of Assumption \ref{assump:leafedges} ensures that no two leaves will be labeled with identical sequences as sequence length $k$ becomes long.

%%%%%%%%%%%%%%%%%

\subsection{Regularized estimators for phylogenetic inference}
\label{sec:regDefs}
Throughout the paper, we consider regularization-type estimators, which are defined as the minimizer of the phylogenetic likelihood function penalized with various $R_k$:
%JF is this argmin unique? I thought the problem was not convex?
%V It's correct that the argmin may not be unique. For very technical contexts, I think the proper notation is 'q^{k, R_k} \in  \argmin_...', but that's a bit too fancy.
\begin{equation}
q^{k, R_k}=  \argmin_{q \in \mathcal{T}}{\, - \frac{1}{k}\ell_k(q) +  \lambda_k R_k(q)}.
\label{eq:objective}
\end{equation}
Here $R_k$ denotes the penalty function and $\lambda_k$ is the regularization parameter that controls how the penalty function impacts the estimates.
Different forms of the penalty function will lead to different statistical estimators of the generating tree.

The existence of a minimizer as in $\eqref{eq:objective}$ is guaranteed by the following Lemma (proof in the Appendix):
\begin{restatable}{Lemma}{REexist}
If the penalty $R_k$ is continuous on $\mathcal{T}$, then for $\lambda>0$ and observed sequences $\mb{Y}^k$, there exists a $q \in \mathcal{T}$ minimizing
\[
Z_{\lambda, \mb{Y}^k}(q) = - \frac{1}{k}\ell_k(q)  + \lambda R_k(q).
\]
\label{lem:exist}
\end{restatable}

We are especially interested in the ability of the estimators to detect polytomies and sampled ancestors.
This leads us the following definition of topological consistency, which in the usual variable selection setting is sometimes called \emph{sparsistency}.

\begin{Definition}
For any vector of branch lengths $q$, we denote the index set of zero entries with
\[
\mathcal{A}(q)=\{i: q_i = 0\}.
\]
We say a regularized estimator with penalty function $R_k$ is \emph{topologically consistent} if for all data-generating branch lengths $q^*$, we have
\[
\lim_{k \rightarrow \infty} {\mathbb{P}\left(\mathcal{A}(q^{k, R_k})=\mathcal{A}(q^*)\right)} =1.
\]
\end{Definition}

\begin{Definition}[Phylogenetic LASSO]
The phylogenetic LASSO estimator is \eqref{eq:objective} with the standard LASSO penalty $R_k^{[0]}$, which in our setting of non-negative $q_i$ is
\[
R_k^{[0]} (q) =\sum_{i \in E} q_i.
\]
We will use $q^{k, R_k^{[0]}}$ to denote the phylogenetic LASSO estimate, namely
\[
q^{k, R_k^{[0]}} =  \arg\min_{q \in \mathcal{T}}{- \frac{1}{k}\ell_k(q) +  \lambda_k^{[0]} \sum_{i \in E}{q_i}}.
\]
\end{Definition}

In general, the classical LASSO may not be topologically consistent, since the $\ell_1$ penalty forces the coefficients (in our case, the branch lengths) to be equally penalized \citep{zou2006adaptive}.
To address this issue, one may assign different weights to different branch lengths by first constructing a naive estimate of the branch lengths, then using this initial estimate to design the weights of the penalties.
Such an ``adaptive'' procedure is exactly the idea of adaptive LASSO, as follows.

\begin{Definition}[Adaptive LASSO \citep{zou2006adaptive}]
The phylogenetic adaptive LASSO estimator is \eqref{eq:objective} with penalty function
\[
R_k^{[1]} (q) =  \sum_{i \in E}{w_{k,i} \, q_i} \h \text{where} \h
w_{k,i} = \left(q^{k, R_k^{[0]}}_i \right)^{-\gamma}
\]
for some $\gamma>0$ and $q^{k, R_k^{[0]}}$ is the phylogenetic LASSO estimate.
The regularizing parameter of adaptive LASSO estimator will be denoted by $ \lambda_k^{[1]}$.
Here, we use the convention that $\infty \cdot 0 = 0$, which means that zero branch lengths contribute nothing to the penalty.
\end{Definition}

A reviewer has pointed out a nice connection between the Adaptive LASSO objective and that of weighted least squares phylogenetics.
In weighted least squares phylogenetics, one finds branch lengths and tree topology that minimize the sum of weighted squared differences between a given set of molecular sequence distances $D_{i,j}$ and inferred distances $d_{i,j}$ on a phylogenetic tree.
\citet{Fitch1967-ms} propose that the these squared distances should be weighted by $1/D_{i,j}^2$, while \citet{Beyer1974-hv} propose weighting with $1/D_{i,j}$.
These are structurally similar to our definition of phylogenetic adaptive LASSO for $\gamma=2$ and $\gamma=1$, respectively, although in our hands these terms are penalties rather than the primary objective function.

\begin{Definition}[Multiple-step adaptive LASSO \citep{Buhlmann2008-bi}]
The phylogenetic multiple-step LASSO is defined recursively with the phylogenetic LASSO estimator as the base case $(m=1)$, and the penalty function in \eqref{eq:objective} at step $m$ being
\[
R^{[m]}_k(q) = \sum_{i \in E}{w_{k,i} \, q_i} \h \text{where} \h
w_{k,i} = \left(q^{k, R_k^{[m-1]}}_{i}\right)^{- \gamma},
\]
where $\gamma>0$ and $q^{k, R_k^{[m-1]}}$ is the $(m-1)$-step regularized estimator with penalty function $R^{[m-1]}_k(q)$.
The regularizing parameter of the $m$-step adaptive LASSO estimator will be denoted by $ \lambda_k^{[m]}$.
Again, we use the convention that $\infty \cdot 0 = 0$.
\end{Definition}

In this paper, we aim to prove that if the weights are data-dependent and cleverly chosen, then the estimators are topologically consistent.
Our proof design relies on the parameter $\gamma$, which dictates how strongly we penalize small edges in the estimation.
In Section 3, we show that if $\gamma$ is sufficiently large ($\gamma>\beta-1$, where $\beta$ is a constant that depends on the structure of the problem), the corresponding LASSO procedures are topologically consistent.

\subsection{Related work}

There is a large literature on penalized M-estimators with possibly non-convex loss or penalty functions from both theoretical and numerical perspectives.
The optimization of such estimators has its own rich literature:

\begin{itemize}
\item In the context of least squares and convex regression with non-convex penalties, several numerical procedures have been proposed, including local quadratic approximation (LQA) \citep{fan2001variable}, the minorize-maximize (MM) algorithm \citep{hunter2005variable}, local linear approximation (LLA) \citep{zou2008one}, the concave-convex procedure (CCP) \citep{kim2008smoothly} and coordinate descent \citep{kim2008smoothly, mazumder2011sparsenet}.
\cite{zhang2012general} provided statistical guarantees for global optima of least\hyp squares linear regression with non-convex penalties and showed that gradient descent starting from a LASSO solution would terminate in specific local minima.
\cite{fan2014strong} proved for convex losses that the LLA algorithm initialized with a LASSO solution attains a local solution with oracle statistical properties.
%Chen and Gu (2014) showed that specific local optima of non-convex regularized least-squares problems are stable, so optimization algorithms initialized sufficiently close by will converge to the same optima.
\cite{wang2013calibrating} proposed a calibrated concave-convex procedure that can achieve the oracle estimator.

\item To enable these analyses, various sufficient conditions for the success of $\ell_1$-relaxations have been proposed, including restricted eigenvalue conditions \citep{bickel2009simultaneous, meinshausen2009lasso} and the restricted Riesz property \citep{zhang2008sparsity}.
\cite{pan2015relaxed} also provide results showing that under restricted eigenvalue assumptions, a certain class of non-convex penalties yield estimates that are consistent in $\ell_2$-norm.

\item For studies of regularized estimators with non-convex losses, one prominent approach is to impose a weaker condition known as restricted strong convexity on the empirical loss function, which involves a lower bound on the remainder in the first-order Taylor expansion of the loss function \citep{negahban2012unified, agarwal2010fast, loh2011high, loh2013regularized, loh2017support}.
\end{itemize}

Outside of the optimization framework, previous work has developed regularized procedures aiming at support recovery and model selection.
The goal of this research is to identify the support (the non-zero components) of the data-generating vector of parameters.
\cite{meinshausen2006high} and \cite{zhao2006model} prove that the Irrepresentable Conditions are almost necessary and sufficient for LASSO to select the true model, which provides a foundation for applications of LASSO for feature selection and sparse representation.
Under a sparse Riesz condition on the correlation of design variables, \cite{zhang2008sparsity} prove that the LASSO selects a model of the correct order of dimensionality and selects all coefficients of greater order than the bias of the selected model.
\cite{zou2006adaptive} introduces the adaptive LASSO algorithm, which produces a topologically consistent estimate of the support even in cases when LASSO may not be consistent.
\cite{loh2017support} also show that for certain non-convex optimization problems, under the restricted strong convexity and a beta-min condition (which provides a lower bound on the minimum signal strength), support recovery consistency may be guaranteed.

Our approach for phylogenetic LASSO is inspired by previous work of \citet{zou2006adaptive} and \citet{Buhlmann2008-bi} who carefully choose the weights of the penalty function.
To enable the theoretical analyses of the constructed estimators, we derive a new condition that is similar to the Restricted Strong Convexity condition \citep{loh2013regularized, loh2017statistical}.
However, instead of imposing regularity conditions directly on the empirical log-likelihood function, we use concentration arguments to analyze the empirical log-likelihood function through its expectation.

We note that penalized likelihood has appeared before in phylogenetics \citep{Kim2008-yd,dinh2018consistency}, although we believe ours to be the first application of a LASSO-type penalty on phylogenetic branch lengths.

\section{Theoretical properties of LASSO-type regularized estimators for phylogenetic inference}

We next show convergence and topological consistency of the LASSO-type phylogenetic estimates introduced in the previous section.
As described in the introduction, phylogenetic LASSO is a non-convex regularization problem for which the true estimates lie on the boundary of a space on which the likelihood function is untamed.
To circumvent those problems, we take a minor departure from the standard approach for analysis of non-convex regularization: instead of imposing regularity conditions directly on the empirical log-likelihood function, we investigate the expected per-site log likelihood and investigate its regularity.
This function enables us to isolate the singular points and derive a local regularity condition that is similar to the Restricted Strong Convexity condition \citep{loh2013regularized, loh2017statistical}.
This leads us to study the fast-rate generalization of the empirical log-likelihood in a PAC learning framework \citep{van2015fast, dinh2016fast}.

\subsection{Definitions and lemmas}

We begin by setting the stage with needed definitions and lemmas.
All proofs have been deferred to the Appendix.
\begin{Definition}
We define the \emph{expected per-site log-likelihood}
\[
\phi(q) : = \mathbb{E}_{\psi \sim P_{q^*}}[\log P_{q}(\psi)]
\]
for any vector of branch lengths $q$.
\end{Definition}

\begin{Definition}
For any $\mu>0$, we denote by $\mathcal{T}(\mu)$ the set of all branch length vectors $q \in \mathcal{T}$ such that $\log P_{q}(\psi) \ge - \mu$ for \emph{all} state assignments $\psi$ to the leaves.
\end{Definition}

We have the following result, where $\| \cdot \|_2$ is the $\ell_2$-norm in $\mathbb{R}^{2N-3}$.
\begin{Lemma}[Limit likelihood]
The vector $q^*$ is the unique maximizer of $\phi$, and $\forall q \in \mathcal{T}$
\begin{equation}
    \frac{1}{k} \ell_k(q)\to \phi(q) \h \text{almost surely}.
\label{eqn:converge}
\end{equation}
%JF Why beta \ge 2? (I'm asking out of curiousity. I didn't know about this inequality before but it looks immensely useful. Up to now, I've assumed such things!) So this is stating that you have local strong convexity! Super neat and so so useful!
%V One way to look at this is through the Taylor expansion at $q^*$. Since the derivative at $q^*$ is zero, the first-order term in the expansion is 0. On the other hand, the high-order derivatives is small, so we have f(q) - f(q^*) ~ (q-q^*)^t H (q-q^*). If the Hession H is positive definite, then we have the inequality, which corresponds to beta=2.
Moreover, there exist $\beta\ge 2$ and $c_1>0$ depending on $N, Q, \eta, g_0, \mu$ such that
\begin{equation}
c_1^{\beta}\|q - q^*\|_2^{\beta} \le |\phi(q) - \phi(q^*)| \h \forall q \in \mathcal{T}(\mu).
\label{loja}
\end{equation}
\label{lem:limlink}
\end{Lemma}

\begin{proof}
The first statement follows from the identifiability assumption, and \eqref{eqn:converge} is a direct consequence of the Law of Large Numbers.
Equation \ref{loja} follows from the Lojasiewicz inequality \citep{ji1992global} for $\phi$ on $\mathcal{T}$, which applies because $\phi$ is an analytic function defined on the compact set $\mathcal{T}$ with $q^*$ as its only maximizer in $\mathcal{T}$.
\end{proof}

Group-based DNA sequence evolution models are a class of relatively simple models that have transition matrix structure compatible with an algebraic group~\citep{Evans1993-jg}.
From Lemma 6.1 of \citet{dinh2018consistency}, we have
\begin{Remark}
For group-based models, we can take $\beta=2$.
\end{Remark}

For any $\mu>0$, we also have the following estimates showing local Lipschitzness of the log-likelihood functions, recalling that $k$ is the number of sites.

\begin{restatable}{Lemma}{RElocallips}
For any $\mu>0$, there exists a constant $c_2(N, Q, \eta, g_0, \mu) > 0$ such that
\begin{equation}
\left |\frac{1}{k}\ell_k(q) - \frac{1}{k}\ell_k(q')  \right|\le c_2 \|q-q'\|_2
\label{eqn:union1}
\end{equation}
and
\begin{equation}
|\phi(q) - \phi(q')| \le c_2 \|q-q'\|_2
\label{eqn:union2}
\end{equation}
for all $q ,q' \in \mathcal{T}(\mu)$.
\label{lem:locallips}
\end{restatable}

Fix an arbitrary $\mu>0$.
For any $q \in \mathcal{T}(\mu)$ we consider the \emph{excess loss}
\[
U_k(q) = \frac{1}{k}\ell_k(q^*) - \frac{1}{k}\ell_k(q).
\]
and derive a PAC lower bound on the deviation of the excess loss from its expected value on $\mathcal{T}(\mu)$.
First note that since the sites $\mb{Y}^k$ are independent and identically distributed, we have
\[
\mathbb{E}\left [U_{k}(q)\right] = \mathbb{E}\left [\frac{1}{k}\ell_k(q^*) - \frac{1}{k}\ell_k(q)\right]  = \phi(q^*) - \phi(q).
\]
Moreover, from Lemma~\ref{lem:locallips}, we have $|U_{k}(q)| \le c_2\|q-q^*\|_2$.
This implies by \eqref{loja} that
\begin{equation}
|U_{k}(q)| \le c_2\|q-q^*\|_2 \le  \frac{c_2}{c_1}\mathbb{E}\left [U_{k}(q)\right]^{1/\beta}
\label{eq:varExp}
\end{equation}
for all $q \in \mathcal{T}(\mu)$.

%Erick's summary: "if U_k is biggish in expectation, then it's biggish with high probability."
\begin{restatable}{Lemma}{REunifboundOnesided}
Let $G_k$ be the set of all branch length vectors $q \in \mathcal{T}(\mu)$ such that $\mathbb{E}\left [U_{k}(q)\right] \ge 1/k$.
Let $\beta \ge 2$ be the constant in Lemma $\ref{lem:limlink}$.
For any $\delta>0$ and previously specified variables there exists $C(\delta, N, Q, \eta, g_0, \mu, \beta)\ge 1$ (independent of $k$) such that for any $k \ge 3$, we have:
\[
U_k(q) \ge \frac{1}{2} \mathbb{E}[U_k(q)] - \frac{C\log k}{k^{2/\beta}}\h \h \forall q \in G_k
\]
with probability greater than $1-\delta$.
\label{lem:unifbound-onesided}
\end{restatable}

We also need the following preliminary lemma from \citep{dinh2016fast}.

\begin{Lemma}
Given $0<\nu<1$, there exist constants $C_1, C_2>0$ depending only on $\nu$ such that for all $x>0$, if $x \le a x^{\nu} + b$ then $x \le C_1 a^{1/(1-\nu)} + C_2 b$.
\label{lem:boost}
\end{Lemma}

\subsection{Convergence and topological consistency of regularized phylogenetics}

We now show convergence and topological consistency of $q^{k, R_k}$, the regularized estimator \eqref{eq:objective}, for various choices of penalty $R_k$ as the sequence length $k$ increases.
For convenience, we will assume throughout this section that the parameters $N, Q, \eta, g_0, \mu$ and $\beta$ (defined in the previous section) are fixed.

We first have the following two lemmas guaranteeing that if $\mu$ is carefully chosen, a neighborhood $V$ of $q^*$ and the regularized estimator $q^{k, R_k}$ lie inside $\mathcal{T}(\mu)$ with high probability.

\begin{restatable}{Lemma}{REstar}
There exist $\mu^*>0$ and an open neighborhood $V$ of $q^*$ in $\mathcal{T}$ such that $V \subset \mathcal{T}(\mu^*)$.
\label{lem:star}
\end{restatable}

\begin{restatable}{Lemma}{RETmu}
If the sequence $\{\lambda_k R_k(q^*)\}$ is bounded, then for any $\delta>0$, there exist $\mu(\delta) >0$ and $K(\delta)>0$ such that for all $k \ge K$,
$q^{k, R_k} \in \mathcal{T}(\mu)$ with probability at least $1-2\delta$.
\label{lem:Tmu}
\end{restatable}

These results enable us to prove a series of theorems establishing consistency and topological consistency of phylogenetic adaptive and multi-step adaptive LASSO.
As part of this development we will first use as a hypothesis and then establish the technical condition that there exists a $C_3 > 0$ independent of $k$ such that
\begin{equation}
|R_k(q^{k, R_k}) - R_k(q^*) | \le C_3 \|q^{k, R_k}-q^*\|_2 \h \forall k.
\label{eqn:boost}
\end{equation}
This will form an essential part of our recursive proof.
As the first step in this project, choosing $\mu$ to satisfy these lemmas, we can use the deviation bound of Lemma~\ref{lem:unifbound-onesided} to prove
\begin{Theorem}
If $\lambda_k R_k(q^*) \to 0$ then $\{q^{k, R_k}\}$ converges to $ q^*$ almost surely.

Moreover, letting $\beta \ge 2$ be the constant in Lemma $\ref{lem:limlink}$, for any $\delta>0$ there exist $C(\delta) >0$ and $K(\delta)>0$ such that for all $k \ge K$, with probability at least $1-\delta$ we have
\begin{equation}
\| q^{k, R_k} -q^*\|_2 \le C(\delta) \left (\frac{\log k}{ k^{2/\beta}}+\lambda_k R_k(q^*)  \right)^{1/\beta}.
\label{eqn:baseOne}
\end{equation}
If we assume further that there exists a $C_3 > 0$ independent of $k$ satisfying \eqref{eqn:boost} then there exists $C'(\delta) >0$ such that for all $k \ge K$,
\[
\| q^{k, R_k} -q^*\|_2 \le C'(\delta) \left (\frac{\log k}{ k^{2/\beta}}+\lambda_k^{\beta/(\beta-1)}  \right)^{1/\beta}
\]
with probability at least $1-\delta$.
\label{theo:base}
\end{Theorem}

Another goal of this section is to prove that the phylogenetic LASSO is able to detect zero edges, which then give polytomies and sampled ancestors.
Since the estimators are defined recursively, we will establish these properties of adaptive and multi-step phylogenetic LASSO through an inductive argument.
Throughout this section, we will continue to use $q^{k, R_k}$ to denote the regularized estimator \eqref{eq:objective}.
We will use $q^{k, S_k}$ to denote the corresponding adaptive estimator where $S_k(q) = \sum_{i}{w_{k,i} \, q_i}$ and $w_{k,i} = \left(q^{k, R_k}_i\right)^{-\gamma}$ for some $\gamma>0$.
We will use $\alpha_k$ to be the regularizing parameter for the second step (regularizing with $S_k$) and keep $\lambda_k$ as the parameter for the first step.
These two need not be equal.

For positive sequences $f_k, g_k$, we will use the notation $f_k  \succ g_k$ to mean that $\lim_{k \to \infty}{f_k/g_k} = \infty$.
We have the following result showing consistency of adaptive LASSO, and setting the stage to show topological consistency of adaptive LASSO.

\begin{Theorem}
Assume that $\lambda_k \to 0$, $R_k(q^*) = \mathcal{O}(1) $ and that
\[
\alpha_k \to 0, \h \h \alpha_k \succ  \left( \frac{\log k}{k^{2/\beta}}\right)^{\gamma/\beta} , \h \h \alpha_k \succ \lambda_k^{\gamma/(\beta-1)}.
\]
We have
\begin{itemize}
\item[(i)] $S_k(q^*) = \mathcal{O}(1)$ and the estimator $q_k^S$ is consistent.
\item[(ii)] If there exists $C_3$ independent of $k$ satisfying \eqref{eqn:boost} then the estimator $q_k^S$ is topologically consistent.
\end{itemize}
\label{theo:adapt}
\end{Theorem}

We also obtain the following Lemma, which proves the regularity of the multiple-step adaptive LASSO, as describe by Equation $\eqref{eqn:boost}$:

\begin{Lemma}
If $q^{k, S_k}$ is topologically consistent and $q^{k, R_k}$ is consistent, then there exists a $C_3$ independent of $k$ such that
$$
|S_k(q^{k, S_k}) - S_k(q^*) | \le C_3 \|q^{k, S_k}-q^*\|_2 \h \forall k.
$$
\label{lem:lips}
\end{Lemma}

This recursive regularity condition helps establish the main result:

\begin{Theorem}
If
\[
\lambda_k^{[m]} \to 0, \h  \lambda_k^{[m]} \succ  \left( \frac{\log k}{k^{2/\beta}}\right)^{\gamma/\beta}, \h  \forall m=0, \ldots, M
\]
and
\begin{equation}
\lambda_k^{[m]}  \succ \left(\lambda_k^{[m-1]}\right)^{\gamma/(\beta-1)} \h  \forall m=1, \ldots, M
\label{eq:lambda}
\end{equation}
then
\begin{itemize}
\item[(i) ] The adaptive LASSO and the $m$-step LASSO are topologically consistent for all $1 \le m \le M$.
\item[(ii) ] For all $0 \le m \le M$, the $m$-step LASSO (including the phylogenetic LASSO and adaptive LASSO) are consistent.
Moreover, for all $\delta>0$ and $0 \le m \le M$, there exists $C^{[m]}(\delta) >0$ such that for all $k \ge K$,
\[
\| q^{k, R^{[m]}_k} -q^*\|_2 \le C^{[m]}(\delta) \left (\frac{\log k}{ k^{2/\beta}}+\left(\lambda_k^{[m]}\right)^{\beta/(\beta-1)}  \right)^{1/\beta}
\]
with probability at least $1-\delta$.
In other words, the convergence of $m$-step LASSO is of order
\[
\mathcal{O}_{P}\left ( \left(\frac{\log k}{ k^{2/\beta}}+\left(\lambda_k^{[m]}\right)^{\beta/(\beta-1)}   \right)^{1/\beta}\right)
\]
where $\mathcal{O}_{P}$ denotes big-O-in-probability.
\end{itemize}
\label{theo:convergence}
\end{Theorem}

\begin{Remark}
If we further assume that $\gamma>\beta -1$, then the results of Theorem $\ref{theo:convergence}$ are valid if $\lambda_k^{[m]}$ is independent of $m$.
This enables us to keep the regularizing parameters $\lambda_k$ unchanged through successive applications of the multi-step estimator.

Similarly, the Theorem applies if  $\gamma>\beta -1$ and
\[
\lambda_k^{[m]}/\lambda_k^{[m-1]} \to c^{[m]}>0
\]
for all $m=1, \ldots, M$.
\end{Remark}

\begin{Remark}
Consider the case $\beta=2$ (for example, for group-based models), $\epsilon>0$ and $\gamma>1$.
If we choose $\lambda_k^{[m]} = \lambda_k$ (independent of $m$) such that
\[
\lambda_k  \sim  \frac{(\log k)^{1/2 + \epsilon}}{\sqrt{k}},
\]
then the convergence of $m$-step LASSO is of order
\[
\mathcal{O}_{P}\left ( \frac{(\log k)^{1/2 + \epsilon}}{\sqrt{k}}\right).
\]
\end{Remark}

We further note that for group-based models, we can take $\beta=2$, and that the theoretical results derived in this section apply for $\gamma>1$.
The limit case when $\gamma=1$ is interesting, for which we believe that the results still hold.
However, the techniques we employ in our framework, including the recursive arguments in Theorem $\ref{theo:convergence}$ and the concentration argument (Lemma \ref{lem:unifbound-onesided}), cannot be adapted to resolve the case.
This issue arises from the fact that less is known about the empirical phylogenetic likelihoods than about their counterparts in classical statistical analyses, which forces us to investigate them indirectly through a concentration argument.

\section{Algorithms}
In this section, we aim to design a robust solver for the phylogenetic LASSO problem.
Many efficient algorithms have been proposed for the LASSO minimization problem
\begin{equation}\label{eq:compOpt}
\hat{q} = \arg\min_{q} g(q) + \lambda \|q\|_1
\end{equation}
for a variety of objective functions $g$.
Note that we now drop the subscript $k$ denoting sequence length from $\lambda$, as we now consider a fixed data set in contrast to the previous theoretical analysis.
When $g(q)=\|Y-Xq\|_2^2$, \citet{efron04} introduced \emph{least angle regression} (LARS) that computes not only the estimates but also the solution path efficiently.
In more general settings, \emph{iterative shrinkage-thresholding algorithm} (ISTA) is a typical proximal gradient method that utilizes an efficient and sparsity-promoting proximal mapping operator (also known as soft-thresholding operator) in each iteration.
Adopting Nesterov's acceleration technique, \citet{Beck2009-ty} proposed a fast ISTA (FISTA) that has been proved to significantly improve the convergence rate.

These previous algorithms do not directly apply to phylogenetic LASSO.
LARS is mainly designed for regression and does not apply here.
Classical proximal gradient methods are not directly applicable for the phylogenetic LASSO for the following reasons:
(i) \emph{Nonconvexity}.
The negative log phylogenetic likelihood is usually non-convex.
Therefore, the convergence analysis (which is described briefly in the following section \ref{sec:pg}) may not hold.
Moreover, nonconvexity also makes it much harder to adapt to local smoothness which could lead to slow convergence.
(ii) \emph{Bounded domain}.
ISTA and FISTA also assume there are no constraints while in phylogenetic inference we need the branches to be nonnegative: $q\geq 0$.
(iii) \emph{Regions of infinite cost}.
Unlike normal cost functions, the negative phylogenetic log-likelihood can be infinite especially when $q$ is sparse as shown in the following proposition.

\begin{Proposition}
Let $\mb{Y}=(y_1,y_2,\ldots,y_N)\in\Omega^{N}$ be an observed character vector on one site. If $y_i\neq y_j$ and there is a path $(u_0,u_1),(u_1,u_2),\ldots,(u_s,u_{s+1}),\; u_0=i, u_{s+1}=j$ on the topology $\tau$ such that $q_{u_ku_{k+1}} = 0, \; k=0,\ldots, s$, then
\[
L(\mb{Y}|q) = 0.
\]
\end{Proposition}
\begin{proof}
Let $a$ be any extension of $\mb{Y}$ to the internal nodes. Since $a_{u_0} = y_i \neq y_j = a_{u_{s+1}}$, there must be some $0\leq k \leq s$ such that $a_{u_k}\neq a_{u_{k+1}}\Rightarrow P_{a_{u_k}a_{u_{k+1}}}(q_{u_ku_{k+1}}) = 0$. Therefore,
\[
L(\mb{Y}|q) = \sum_{a}\eta(a_\rho)\prod_{(u,v)\in E}P_{a_ua_v}(q_{uv}) = 0.
\]
\end{proof}

In what follows, we briefly review the proximal gradient methods (ISTA) and their accelerations (FISTA), and provide an extension of FISTA to accommodate the above issues.

\subsection{Proximal Gradient Methods}\label{sec:pg}
Consider the nonsmooth $\ell_1$ regularized problem \eqref{eq:compOpt}.
Gradient descent generally does not work due to non-differentiability of the $\ell_1$ norm.
The key insight of the proximal gradient method is to view the gradient descent update as a minimization of a local linear approximation to $g$ plus a quadratic term.
This suggests the update strategy
\begin{align}
q^{(n+1)} &= \arg\min_{q}\left\{g(q^{(n)})+\langle\nabla g(q^{(n)}),q-q^{(n)}\rangle+\frac1{2t_n}\left\|q-q^{(n)}\right\|_2^2 + \lambda\|q\|_1\right\}\nonumber\\
&=\arg\min_q\left\{\frac1{2t_n}\left\|q-\left(q^{(n)}-t_n\nabla g(q^{(n)})\right)\right\|_2^2 + \lambda \|q\|_1\right\}\label{eq:pg}
\end{align}
where $t_n$ is the step size.
Note that \eqref{eq:pg} corresponds to the proximal map of $h(q) = \|q\|_1$, which is defined as follows
\begin{equation}\label{eq:prox}
\mathbf{prox}_{th}(p) := \arg\min_{q}\left\{\frac1{2}\|q-p\|_2^2+th(q)\right\} = \arg\min_{q}\left\{\frac1{2t}\|q-p\|_2^2+h(q)\right\}
\end{equation}
If the regularization function $h$ is simple, \eqref{eq:prox} is usually easy to solve.
For example, in case of $h(q)=\|q\|_1$, it can be solved by the \emph{soft thresholding operator}
\[
\mathcal{S}_t(p) = \mathrm{sign}(p)(|p|-t)_{+}
\]
where $x_+=\max\{x,0\}$.
Applying this operator to \eqref{eq:pg}, we get the ISTA update formula
\begin{equation}\label{eq:ista}
q^{(n+1)} = \mathcal{S}_{\lambda t_n}(q^{(n)}-t_n\nabla g(q^{(n)})).
\end{equation}
Let $f=g+\lambda\|q\|_1$.
% JF Double checking -- your function isn't convex though? I have referenced Gong et al. (2013) for solving nonconvex loss functions with lasso via proximal gradient descent.
Assume $g$ is convex and $\nabla g$ is Lipschitz continuous with Lipschitz constant $L_{\nabla g}>0$; if a constant step size is used and $t_n = t < 1/L_{\nabla g}$, then ISTA converges at rate
\begin{equation}\label{eq:rateista}
f(q^{(n)}) - f(q^\ast) \leq \frac{1}{2tn} \|q^{(0)}-q^\ast\|_2^2
\end{equation}
where $q^\ast$ is the optimal solution. This means ISTA has \emph{sublinear convergence} whenever the stepsize is in the interval $(0,1/L_{\nabla g}]$. Note that ISTA could have linear convergence if $g$ is strongly convex.
\vskip3pt
The convergence rate in \eqref{eq:rateista} can be significantly improved using Nesterov's acceleration technique.
The acceleration comes from a weighted combination of the current and previous gradient directions, which is similar to gradient descent with momentum.
This leads to Algorithm \ref{alg:fista} which is essentially equivalent to the \emph{fast iterative shrinkage-thresholding algorithm} (FISTA) introduced by \citet{Beck2009-ty}.
Under the same condition, FISTA enjoys a significantly faster convergence rate
\begin{equation}\label{eq:ratefista}
f(q^{(n)}) - f(q^\ast) \leq \frac{2}{t(n+1)^2}\|q^{(0)}-q^\ast\|_2^2.
\end{equation}
Notice that the above convergence rates both require the stepsize $t\leq 1/L_{\nabla g}$. In practice, however, the Lipschitz coefficient $L_{\nabla g}$ is usually unavailable and \emph{backtracking line search} is commonly used.

\begin{algorithm}\caption{Fast Iterative Shrinkage-Thresholding Algorithm (FISTA)}\label{alg:fista}
\begin{algorithmic}[1]
\Input initial value $q^{(0)}$, step size $t$, regularization coefficient $\lambda$
\State Set $q^{(-1)} = q^{(0)}, n = 1$
\While{not converged}
  \State $p \gets q^{(n-1)} + \frac{n-2}{n+1}(q^{(n-1)}-q^{(n-2)})$ \Comment{Nesterov's Acceleration}
  \State $q^{(n)}	\gets \mathcal{S}_{\lambda t}(p-t\nabla g(p))$ \Comment{Soft-Thresholding Operator}
  \State $n\gets n+1$
\EndWhile
\Output $q^\ast \gets q^{(n)}$
\end{algorithmic}
\end{algorithm}

\subsection{Projected FISTA}
FISTA usually assumes no constraints for the parameters.
However, in the phylogenetic case branch lengths are must be non-negative ($q\geq 0$).
To address this issue, we combine the projected gradient method (which can be viewed as proximal gradient as well) with FISTA to assure non-negative updates.
We refer to this hybrid as \emph{projected} FISTA (pFISTA).
Note that a similar strategy has been adopted by \citet{liu15} in tight frames based magnetic resonance image reconstruction.
Let $\mathcal{C}$ be a convex feasible set, define the indicator function $I_\mathcal{C}$ of the set $\mathcal{C}$:
\[
I_\mathcal{C}(q) = \left\{\begin{array}{ll} 0 &\mathrm{if}\; q \in \mathcal{C}, \mathrm{and}\\
+\infty & \mathrm{otherwise}\end{array}\right.
\]
With the constraint $q\in\mathcal{C}$, we consider the following projected proximal gradient update
\begin{align}
q^{(n+1)} &= \arg\min_{q\in\mathcal{C}}\left\{g(p)+\langle\nabla g(p),q-p\rangle+\frac1{2t_n}\|q-p\|_2^2 + h(q)\right\}\nonumber\\
&=\arg\min_q\left\{g(p)+\langle\nabla g(p),q-p\rangle+\frac1{2t_n}\|q-p\|_2^2 + h(q) + I_\mathcal{C}(q)\right\}\nonumber\\
&=\mathbf{prox}_{t_nh_\mathcal{C}}(p-t_n\nabla g(p))\label{eq:ppg}
\end{align}
where $h_\mathcal{C} = h(q) + I_\mathcal{C}(q)$. Using \emph{forward-backward splitting} \citep[see][]{combettes06}, \eqref{eq:ppg} can be approximated as
\begin{equation}\label{eq:splitapprox}
\mathbf{prox}_{t_nh_\mathcal{C}}(p-t_n\nabla g(p)) \approx \Pi_{\mathcal{C}}(\mathbf{prox}_{t_nh}(p-t_n\nabla g(p)))
\end{equation}
where $\Pi_\mathcal{C}$ is the Euclidean projection on to $\mathcal{C}$.
When $h(q)=\|q\|_1,\; \mathcal{C}=\{q:q\geq 0\}$, we have the following pFISTA update formula
\[
p=q^{(n)} + \frac{n-1}{n+2}\left(q^{(n)}-q^{(n-1)}\right),\quad q^{(n+1)} = \left[S_{\lambda t_n}(p_+-t_n\nabla g(p_+))\right]_+.
\]
Note that in this case, \eqref{eq:splitapprox} is actually exact.
Similarly, we can easily derive the projected ISTA (pISTA) update formula and we omit it here.

\subsection{Restarting}
To accommodate non-convexity and possible infinities of the phylogenetic cost function, we adopt the restarting technique introduced by \citet{donoghue13} where they used it as a heuristic means of improving the convergence rate of accelerated gradient schemes.
In the phylogenetic case, due to the non-convexity of negative phylogenetic log-likelihood, backtracking line search would fail to adapt to local smoothness which could lead to inefficient small step size.
Moreover, the LASSO penalty will frequently push us into the ``forbidden'' zone $\{q:g(q)=+\infty\}$, especially when there are a lot of short branches.
We therefore adjust the restarting criteria as follows:

\begin{itemize}
\item Small stepsize: restart whenever $t_n$ is less than a restart threshold $\epsilon$.
\item Infinite cost: restart whenever $g(p_+) = +\infty$.
\end{itemize}

Equipping FISTA with projection and adaptive restarting, we obtain an efficient phylogenetic LASSO solver that we summarize in Algorithm \ref{alg:pfista}.

\begin{algorithm}\caption{Projected FISTA with Restarting}\label{alg:pfista}
\begin{algorithmic}[1]
\Input initial $q^{(0)}$, default step size $t$, regularization coefficient $\lambda$, restart threshold $\epsilon$, backtracking line search parameter $\omega\in(0,1)$
\While{not converged}
\State Set $q^{(-1)} = q^{(0)}, \;t_1=t, \;n = 1$
\While{not converged}
  \State $p \gets q^{(n-1)} + \frac{n-2}{n+1}(q^{(n-1)}-q^{(n-2)})$ \Comment{Nesterov's Acceleration}
  \If{$g(p_+) = +\infty$} \Comment{Restarting}
  \State {\bf break the inner loop}
  \EndIf
  \State $t_n \gets t_{n-1}$
  \State Adapt $t_n$ through \emph{backtracking line search} with $\omega$
  \If{$t_n < \epsilon$} \Comment{Restarting}
  \State {\bf break the inner loop}
  \EndIf
  \State $q^{(n)}	\gets \left[\mathcal{S}_{\lambda t_n}(p_+-t_n\nabla g(p_+))\right]_+$ \Comment{Projected Soft-Thresholding Operator}
  \State $n\gets n+1$
\EndWhile
\State Set $q^{(0)}=q^{(n-1)}$
\EndWhile
\Output $q^\ast \gets q^{(n)}$
\end{algorithmic}
\end{algorithm}
\newpage

\begin{Remark}
Note that the adaptive phylogenetic LASSO
\begin{equation}\label{eq:adalasso}
\hat{q}^{S} = \arg\min_q g(q) + \lambda \sum_{j}w_jq_j
\end{equation}
is equivalent to (using $/$ to denote componentwise division)
\[
\tilde{q}^S = \arg\min_q\left\{g(q/w) + \lambda \|q\|_1\right\},\quad \hat{q}^S = \tilde{q}^S/w.
\]
Therefore, Algorithm \ref{alg:pfista} can also be used to solve the (multi-step) adaptive phylogenetic LASSO.
\end{Remark}

\section{Experiments}
In this section, we first demonstrate the efficiency of the proposed algorithm for solving the phylogenetic LASSO problem when combined with maximum-likelihood phylogenetic inference.
We then show (non-adaptive) phylogenetic LASSO does not appear to be strong enough to find zero edges on simulated data; adaptive phylogenetic LASSO performs much better.
We, therefore, compare our adaptive phylogenetic LASSO with simple thresholding and rjMCMC on simulated data and then apply it to some real data sets.
For all simulation and inference, we use the simplest \citet{Jukes1969-hv} model of DNA substitution, in which all substitutions have equal rates.
The choice of regularization $\lambda$ is to a certain extent data dependent.
For the simulated data, we choose a range of $\lambda$s to demonstrate the balance between miss rate and false alarm rate (Figure \ref{fig:consistency}).
For the Dengue virus data, we find that the performance is fairly insensitive to the regularization coefficient once the regularization coefficient is reasonably large (Figure \ref{fig:numzerobranch}).

We use PhyloInfer to compute the phylogenetic likelihood via the pruning algorithm \cite{felsenstein1981evolutionary}, which can be found at \url{https://github.com/zcrabbit/PhyloInfer}.
PhyloInfer is a Python package originally developed for extending Hamiltonian Monte Carlo to Bayesian phylogenetic inference \citep{Dinh2017-oj}.
The code for adaptive phylogenetic LASSO is made available at \url{https://github.com/matsengrp/adaLASSO-phylo}.

\subsection{Efficiency of pFISTA for solving the phylogenetic LASSO}
The fast convergence rate of FISTA (or pFISTA) need not hold when the cost function $g$ is nonconvex.
However, we can expect that $g$ is well approximated by a quadratic function near the optimal (or some local mode) $q^\ast$ \citep{donoghue13}.
That is, there exists a neighborhood of $q^\ast$ inside of which
\[
g(q) \approx g(q^\ast) + \frac12(q-q^\ast)^T\nabla^2g(q^\ast)(q-q^\ast)
\]
When we are eventually inside this domain, we will observe behavior consistent with the convergence analysis in Section~\ref{sec:pg}.

To test the efficiency of pFISTA in different scenarios, we consider various simulated data sets generated from ``sparse'' unrooted trees with $100$ tips and $50$ randomly chosen zero branches as follows.
All simulated data sets contain $1000$ independent observations on the leaf nodes.
We set the minimum step size $\epsilon=5e\text{-}08$ for restarting.

We use the following simulation setups, in which branch lengths are expressed in the traditional units of expected number of substitutions per site.

\begin{center}
\begin{figure}[!t]
\includegraphics[width=0.45\textwidth]{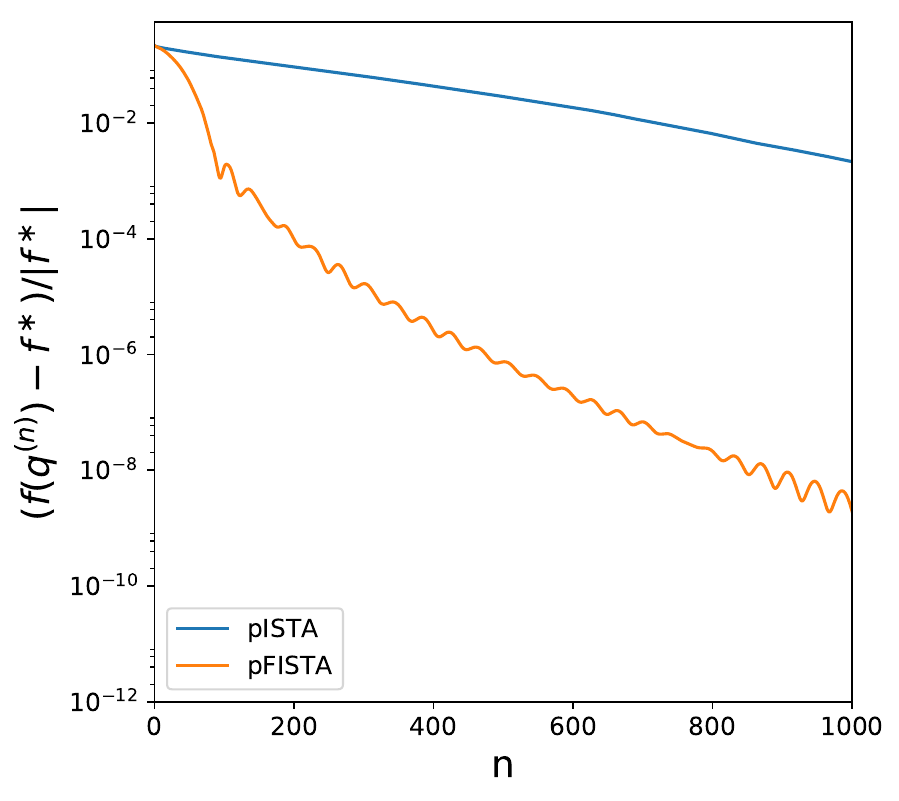}\hspace{20pt}
\includegraphics[width=0.45\textwidth]{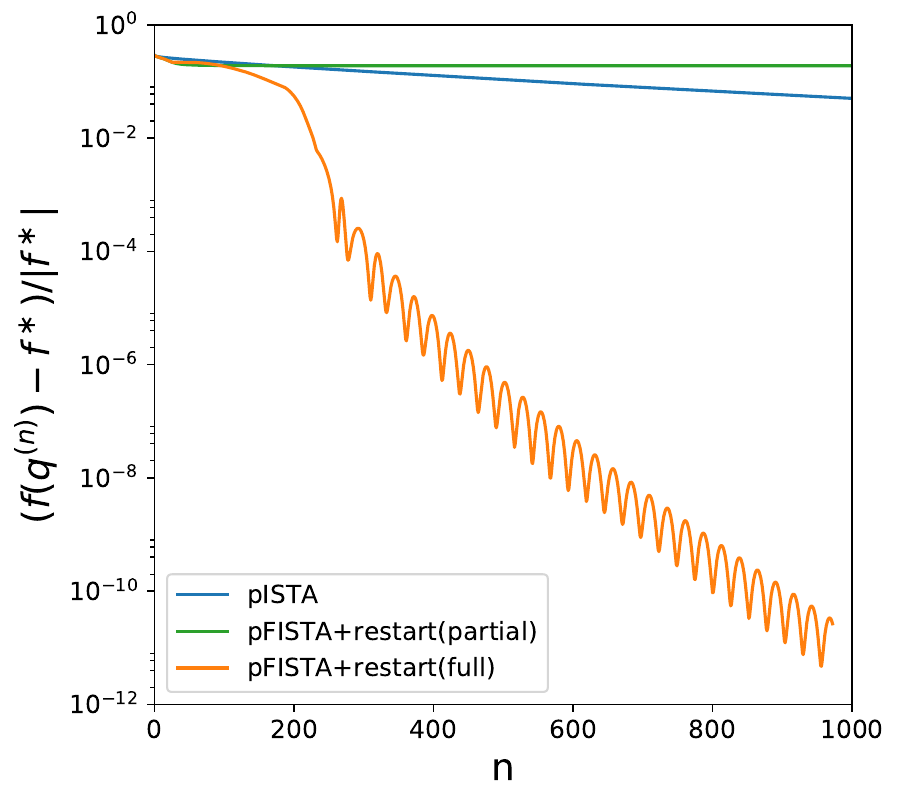}
\caption{
pISTA vs pFISTA on simulated data sets in terms of the relative error $(f(q^{(n)})-f^\ast)/|f^\ast|$, where $f=g+\|q\|_1$ and $f^\ast=f(q^\ast)$. The optimal solution $q^\ast$ is obtained from a long run of pFISTA. We used penalty coefficient $\lambda = 1.0$ for each run.
Left panel: simulation 1; Right panel: simulation 2. In simulation 2, we tried two restarting strategies: restart whenever $g(p_+)=+\infty$ (partial) and restart whenever $t_n<\epsilon$ or $g(p_+)=+\infty$ (full).
}
\label{fig:istavsfista}
\end{figure}
\end{center}

\paragraph{\bf Simulation 1} (\emph{No short branches}).
All nonzero branches have length $0.05$.
Because there are no short nonzero branches, branches that are originally nonzero are less likely to be collapsed to zero and we expect no restarting is needed.

\paragraph{\bf Simulation 2} (\emph{A few short branches}). For all the nonzero branches, we randomly choose $15$ of them and set their lengths to $0.002$. All the other branches have length $0.05$. In this setting, there are a few short branches that are likely to be shrunken to zero. As a result, several restarts may be needed before convergence.
\vskip7pt

We see that when the model does not have very short non-zero branches and the phylogenetic cost is more regular, pFISTA finds the quadratic domain quickly and performs consistently with the corresponding convergence rate in equation \eqref{eq:ratefista}, even without restarting (Figure~\ref{fig:istavsfista}, left).
When the model does have many very short branches and the negative phylogenetic log-likelihood is highly nonconvex, pFISTA with restart still manages to arrive at the quadratic domain quickly and exhibits fast convergence thereafter.
Furthermore, we find the \emph{small stepsize} restarting criterion is useful to adapt to changing local smoothness and facilitate mode exploration.
In both situations, pFISTA performs consistently better than pISTA.
As a matter of fact, pISTA is monotonic so is more likely to get stuck in local minima, and hence may not be suitable for nonconvex optimization.
We, therefore, use pFISTA with restart as our default algorithm in all the following experiments.

\begin{Remark}
Like other non-convex optimization algorithms, pFISTA with restart may be sensitive to the starting position of the parameters. However, due to the momentum introduced in Nesterov's acceleration (which causes the ripples in Figure \ref{fig:istavsfista}) and adaptive restarting, pFISTA with restart is more likely to escape local minima and potentially arrive at the global minimum.
\end{Remark}

\subsection{Performance of phylogenetic LASSO}
Through simulation we also find that in practice the (non-adaptive) phylogenetic LASSO penalty is not strong enough to find all zero branches.
Indeed, we find that phylogenetic LASSO only recovers around 60\% of the sparsity found in the true models and larger penalty does not necessarily give more sparsity (Table \ref{tab:inconsistency}).
This suggests we use the multistep adaptive phylogenetic LASSO that has been proven to be topologically consistent under mild conditions (Theorem~\ref{theo:convergence}).

\begin{table}[ht]
\begin{center}
\begin{tabularx}{0.56\textwidth}{@{}c|ccccccc@{}}
\toprule
$\lambda$  &   $1$     &    $5$   &   $10$   &   $20$   &   $40$    &      $80$    &    $160$   \\
\midrule
Simulation 1      &    $32$  &  $32$   &    $32$  &   $32$   &   $32$    &      $32$    &     $32$    \\
Simulation 2      &    $32$  &  $32$   &    $32$  &   $32$   &   $32$    &      $32$    &     $31$    \\
\bottomrule
\end{tabularx}
\vspace{10pt}
\caption{Number of correct zero length branches found by (non-adaptive) phylogenetic LASSO using various penalty coefficients in both simulation models, each of which have 50 zero length branches.}\label{tab:inconsistency}
\end{center}
\vspace{-10pt}
\end{table}

\subsection{Performance of adaptive phylogenetic LASSO}
Next, we demonstrate that the topologically consistent (multistep) adaptive phylogenetic LASSO significantly enhances sparsity on simulated data compared to phylogenetic LASSO.
We will use the more difficult simulation 2 that have a combination of zero and very short branches.
In what follows (and for the rest of this section), we compute adaptive and multistep adaptive phylogenetic LASSO as described in Section~\ref{sec:regDefs}.
Note that $m=1$ (first cycle) is the phylogenetic LASSO and $m=2$ (second cycle) corresponds to the adaptive phylogenetic LASSO.
Therefore, we can compare all phylogenetic LASSO estimators by simply running the multistep adaptive phylogenetic LASSO with the maximum cycle number $M\geq 2$.
Our theoretical results are for $\gamma>1$, however we have found that in practice large $\gamma$ often leads to severe adaptive weights and hence numerical instability.
Thus we use $\gamma=1$ in the following experiments and put some results for $\gamma>1$ (with guaranteed topological consistency) in the Appendix.

We run the multistep phylogenetic LASSO with $M=4$ cycles.
To test the topological consistency of the estimators, we use different initial regularization coefficients $\lambda^{[0]}=10, 20, 30, 40, 50$ and update the regularization coefficients according to
\[
\lambda^{[m]} = \lambda^{[m-1]} \, \frac{\mathrm{mean}((\hat{q}^{[m-1]})^\gamma)}{\mathrm{mean}((\hat{q}^{[m-2]})^\gamma)}, \quad \hat{q}^{[-1]} = 1
\]
which maintains a relatively stable regularization among the adaptive LASSO steps because the varying part of the regularization is roughly $\lambda^{[m]} / \mathrm{mean}((q^{[m-1]})^\gamma)$ for the $m$th cycle.
This formula provides reasonably good balance between sparsity (identified zero branches) and numerical stability in our experiments.

We find that multistep adaptive phylogenetic LASSO does improve sparsity identification while maintaining a relatively low misidentification rate.
Indeed, as the cycle number increases, the estimator now is able to identify more zero branches (Figure \ref{fig:consistency}, upper left panel).
Moreover, unlike the phylogenetic LASSO ($m=1$), we do observe more sparsity when the regularization coefficient increases at cycles $m>1$.
As more cycles are run and larger penalty coefficients are used, we see that multistep adaptive LASSO manages to reduce miss detection (i.e.\ unidentified zero branches) without introducing many extra false alarms (misidentified zero branches; Figure \ref{fig:consistency}, upper right panel).
In contrast, simple thresholding is more likely to misidentify zero branches when larger thresholds are used to bring down miss detection.
The choice of the regularization coefficient $\lambda$ is also important.
While small $\lambda$ is not enough for detecting most zero branches, these simulations show that a too-large $\lambda$ is likely to increase the number of false detections (Figure \ref{fig:consistency}, bottom panel).
On the other hand, as described below, for real data we find less dependence on the exact value of $\lambda$.

\begin{figure}
\begin{center}
\includegraphics[width=0.45\textwidth]{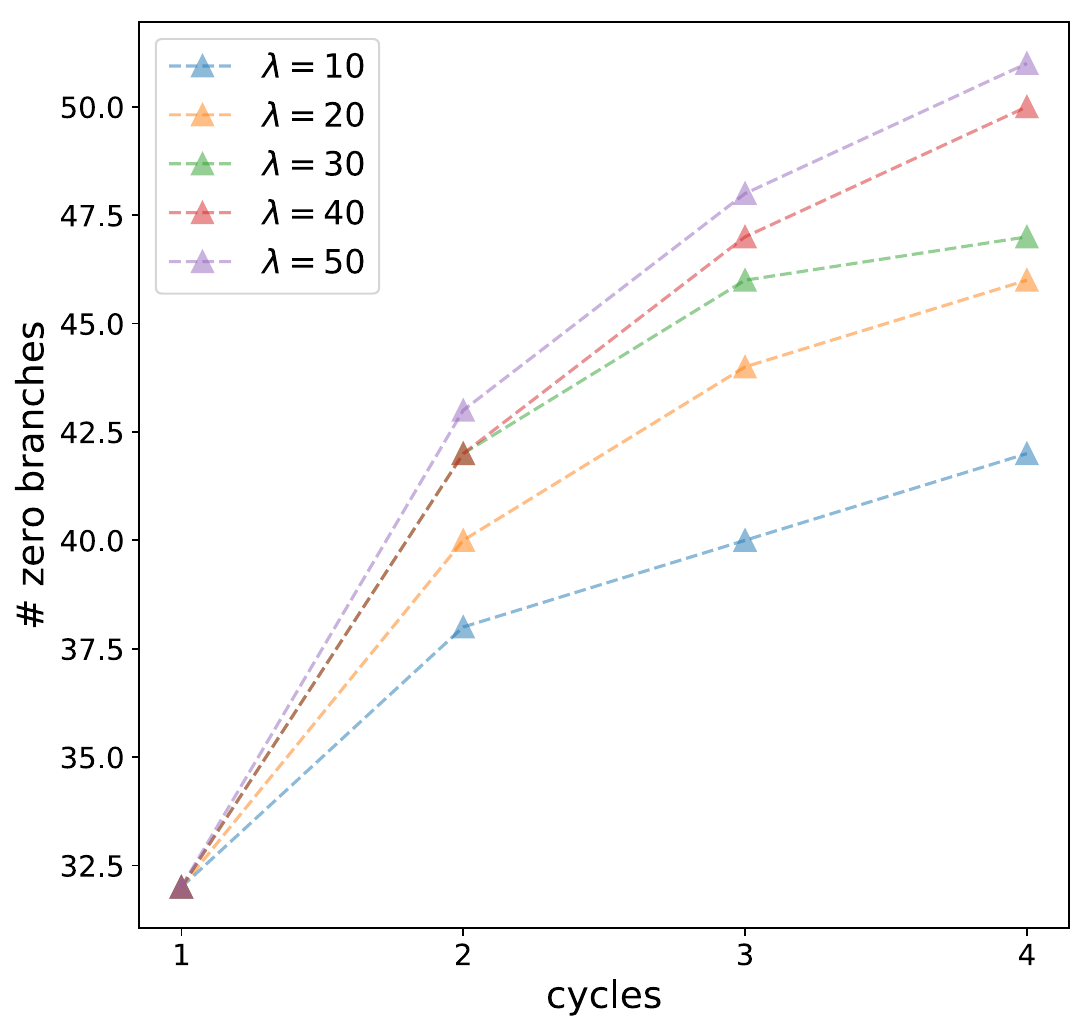}\hspace{20pt}
\includegraphics[width=0.45\textwidth]{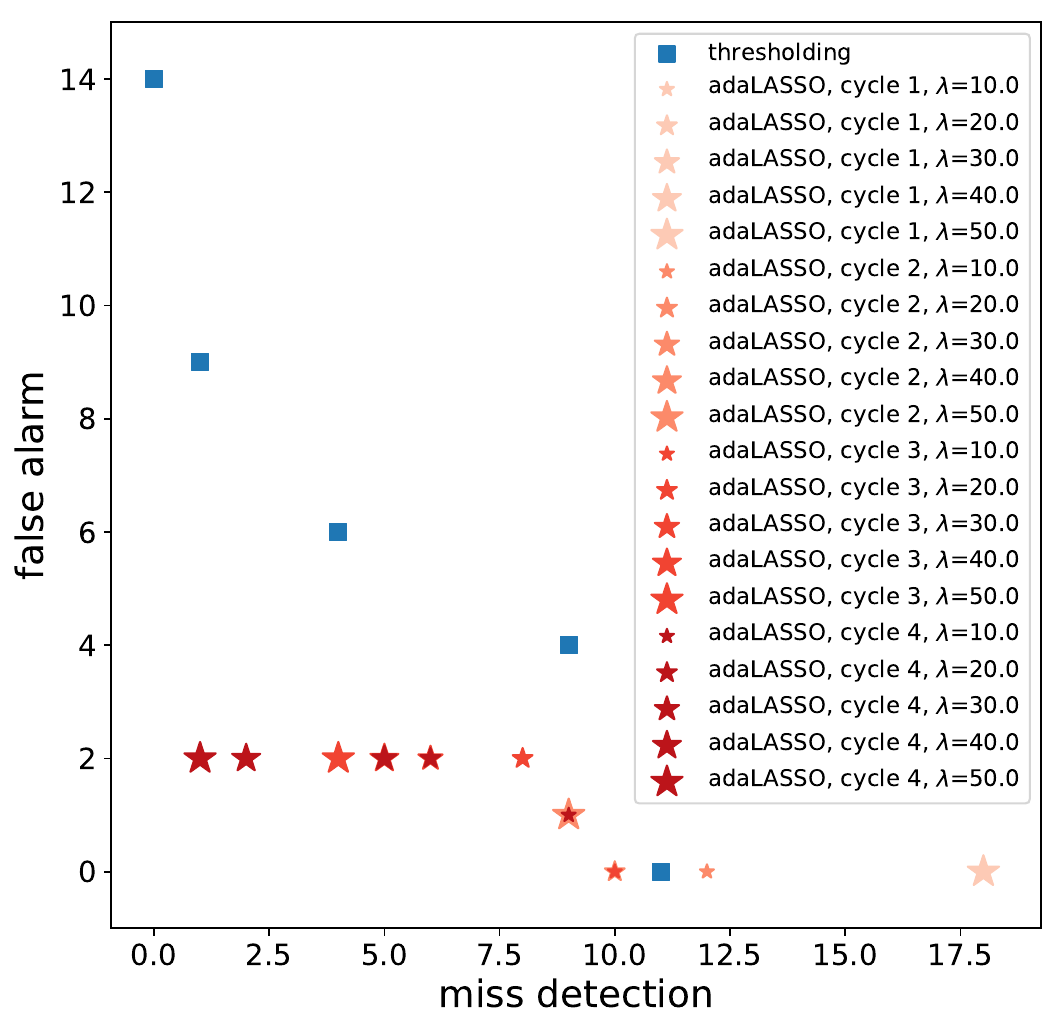}
\includegraphics[width=0.50\textwidth]{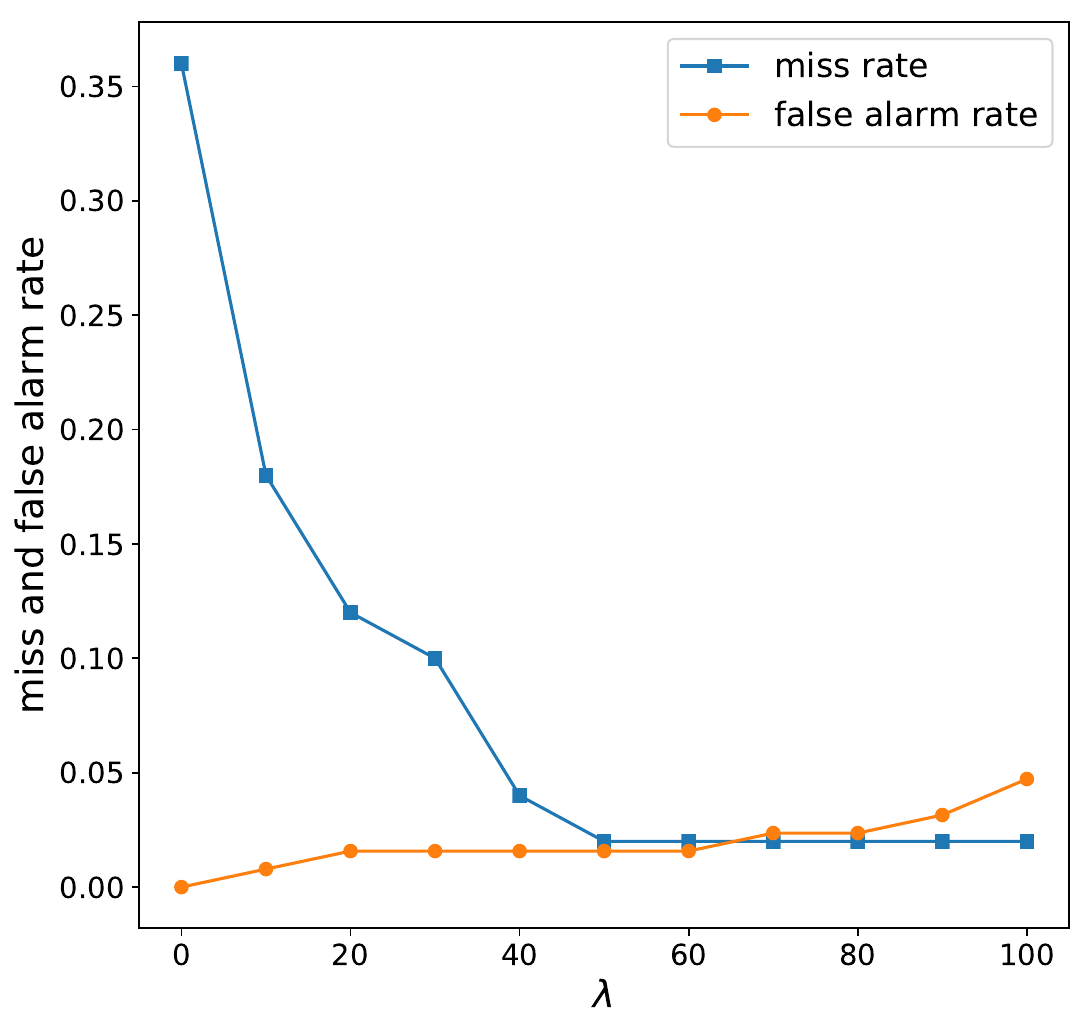}
\end{center}
\caption{Topological consistency comparison of different phylogenetic LASSO procedures on simulation 2. Upper Left panel: number of identified zero branches after various numbers of multistep adaptive LASSO cycles. Upper Right panel: the number of misidentified zero branches (false alarm) and the number of unidentified zero branches (miss detection) for simple thresholding and multistep adaptive phylogenetic LASSO at different cycles. Bottom panel: miss rate and false alarm rate as a function of the regularization coefficient with 4 cycles.}\label{fig:consistency}
\end{figure}

\subsection{Short Edge Detection}
Previous work has proposed Bayesian approaches to infer non-bifurcating tree topologies by assigning priors that cover all of the tree space, including less-resolved tree topologies \citep{Lewis2005-ez,Lewis2015-kv}.
Since the numbers of branches (parameters) are different among those tree topologies, reversible-jump MCMC (rjMCMC) is used for posterior inference.
Both sparsity-promoting priors and adaptive LASSO are means of sparsity encouragement that allow us to discover non-bifurcating tree topologies, which as described in the Introduction make different evolutionary statements than their resolved counterparts.
%By either assigning priors that favor non-bifurcating tree topologies or using $\ell_1$ penalization on branch lengths, rjMCMC and (multistep) adaptive phylogenetic LASSO allow us to discover non-bifurcating tree topologies, which as described in the Introduction make different evolutionary statements than their resolved counterparts.
However, those sparsity encouraging procedures also make it much more difficult to detect relatively short edges.
Thus, we would like to understand the performance of methods in terms of detection probability: the probability of inferring a branch to be of non-zero length.

To investigate how short an edge can be and still be detected by both methods, we follow \citet{Lewis2005-ez} and simulate a series of data sets using the same tree as in Simulation 2.
All branch lengths are the same as in that simulation except those for the $15$ randomly chosen short branches, each of which we take to be $0.0, 0.002, 0.004, 0.006, 0.008, 0.010$ for the various trials; the nonzero short branches are meant to be particularly challenging to distinguish from the actual zero branches.
For each of these six lengths, we simulate $100$ data sets of the same size (1000 sites).
These values for short branches are multiples of $1/1000$, which provides, on average, one mutation per data set along the branch of interest.
Note that branch lengths represent the \emph{expected} number of mutations per site, so for example a branch length of $0.001$ does not guarantee that a mutation will occur on the branch of interest in every simulated data set.
We run multistep adaptive phylogenetic LASSO with $M=4$ cycles and initial regularization coefficient $\lambda^{[0]}=50$, and rjMCMC with the polytomy prior $C=1$ ($C$ is the ratio of prior mass between trees with successive numbers of internal nodes as defined in \citet{Lewis2005-ez}) for analysis.
The detection probabilities of rjMCMC are the averaged split posterior probabilities of the corresponding branches over the 100 independent data sets.

We find that multistep adaptive phylogenetic LASSO indeed strikes a better balance between identifying zero branches and detecting short branches than rjMCMC in this simulation study (Figure \ref{fig:detection}).
In addition to being slightly better at identifying zero branches than rjMCMC (partly due to a weak polytomy prior $C=1$), multistep adaptive phylogenetic LASSO has a substantially improved detection probability for short branches (Figure \ref{fig:detection}).
Also note that sufficiently long branch lengths (about $10$ expected substitution per data set) are needed for an edge to be reliably detected using either method.

\begin{figure}
\begin{center}
\includegraphics[width=0.7\textwidth]{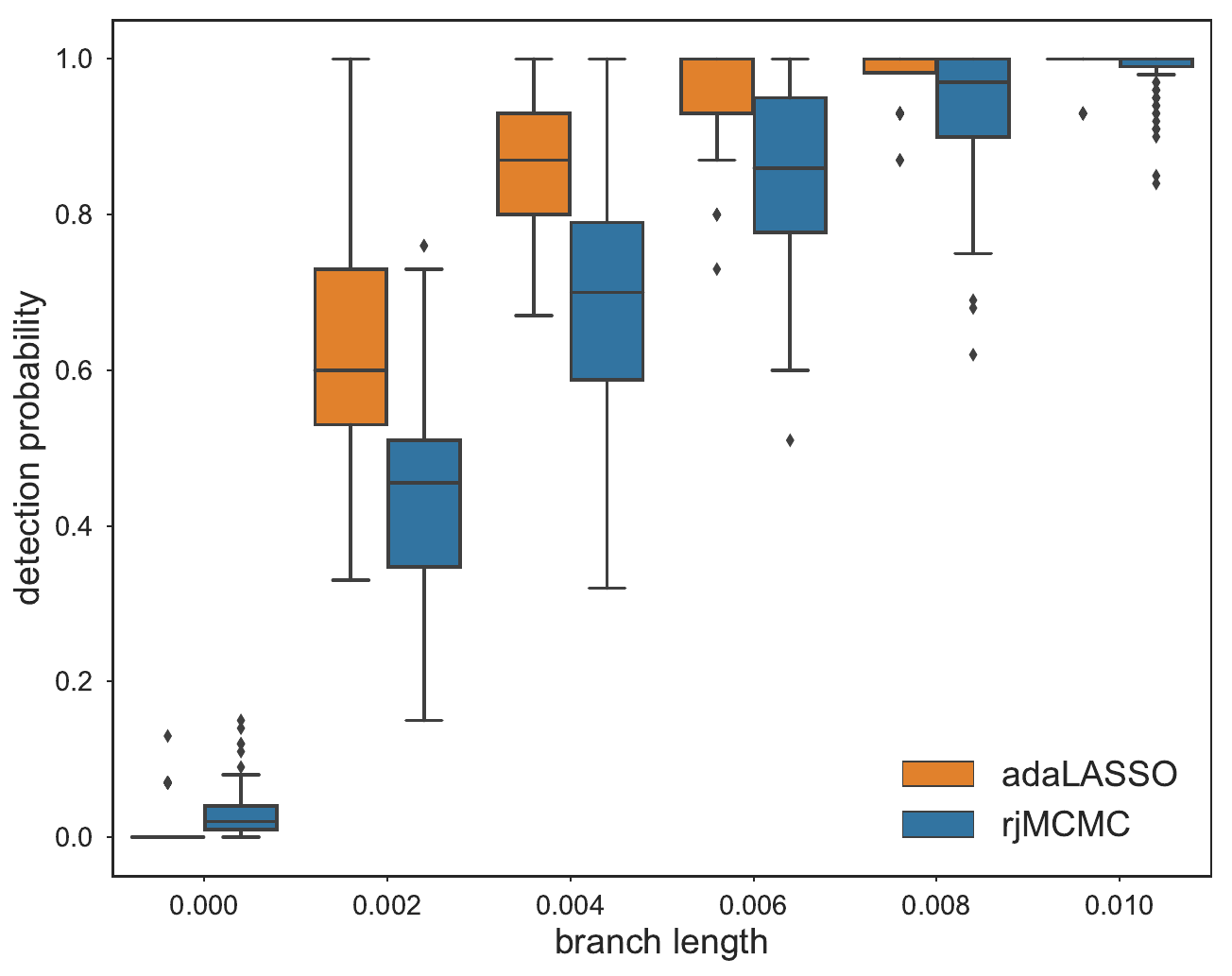}
\end{center}
\caption{Performance of multistep (4 cycle) adaptive phylogenetic LASSO and rjMCMC at detecting short branches.
Detection probability is the probability of inferring a branch to be of non-zero length.
Therefore, the ideal detection probability is 1 for non-zero length branches (all except for the first value on the x-axis) and 0 for zero length branches.
}
\label{fig:detection}
\end{figure}

\subsection{Dengue Virus Data}

We now compare our adaptive phylogenetic LASSO methods to others on a real data set.
So far, we have tested the performance of multistep adaptive phylogenetic LASSO on a fixed topology.
For real data sets, the underlying phylogenies are unknown and hence have to be inferred from the data.
We therefore propose to use multistep adaptive phylogenetic LASSO as a sparsity-enforcing procedure after traditional maximum likelihood based inferences.
In what follows, we use this combined procedure together with bootstrapping to measure edge support on a real data set of the Dengue genome sequences.
In our experiment, we consider one typical subset of the 4th Dengue serotype (``DENV4'') consisting of 22 whole-genome sequences from Brazil curated by the nextstrain project \citep{Hadfield2018-jg} and originally sourced from the LANL hemorrhagic fever virus database \citep{Kuiken2012-vp}.
The sequence alignment of these sequences comprises 10756 nucleotide sites.

Following \citet{Lewis2005-ez}, we conduct our analysis using the following methods: (1) maximum likelihood bootstrapping columns of a sequence alignment; (2) a conventional MCMC Bayesian inference restricted to fully resolved tree topologies; (3) a reversible-jump MCMC method moving among fully resolved as well as polytomous tree topologies; (4) two combined procedures, maximum likelihood bootstrapping plus multistep adaptive phylogenetic LASSO and maximum likelihood bootstrapping plus thresholding, both allow fully bifurcating and non-bifurcating tree topologies.
Maximum likelihood bootstrap analysis is performed using RAxML \citep{stamatakis14} with $1000$ replicates.
The conventional MCMC Bayesian analysis is done in MrBayes \citep{Ronquist2012-hi} where we place a uniform prior on the fully resolved topology and Exponential (rate 10) prior on the branch lengths.
The rjMCMC analysis is run in \texttt{p4} \citep{foster04}, using a flat polytomy prior with $C=1$.
Their code can be found at \url{https://github.com/Anaphory/p4-phylogeny}.
For each Bayesian approach, a single Markov chain was run $8e\text{+}06$ generations after a $2e\text{+}06$ generation burn-in period.
Trees and branch lengths are sampled every $1000$ generations, yielding $8000$ samples.
Both combined procedures are implemented based on the bootstrapped ML trees obtained in (1).
For multistep adaptive phylogenetic LASSO, we use $M=4$ cycles and test different initial regularization coefficients $\lambda^{[0]}=150, 300, 450$.
We set the thresholds $\kappa = 1e\text{-}06,5e\text{-}05,1e\text{-}04$ for the simple thresholding method.

\begin{figure}
\begin{center}
\includegraphics[width=\textwidth]{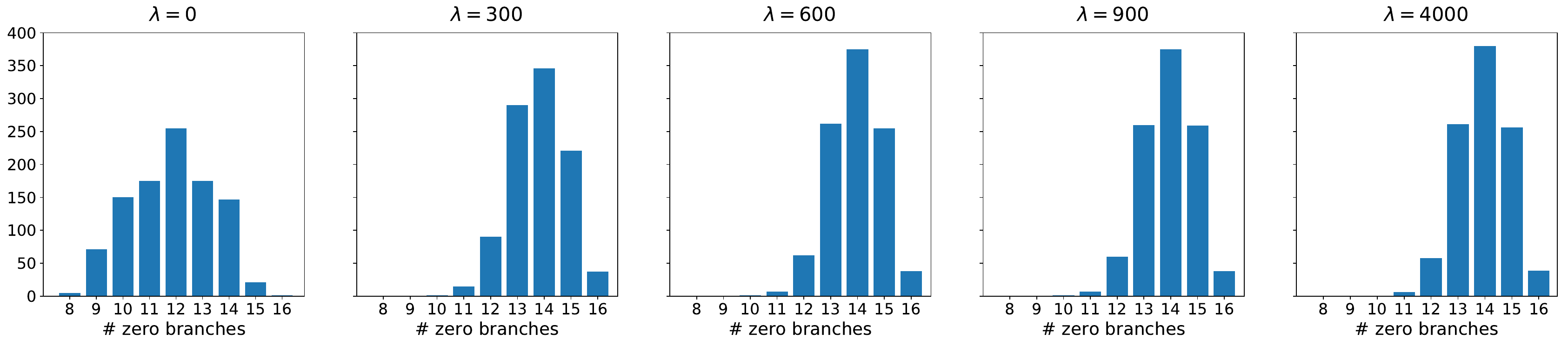}
\caption{The distributions (across bootstrap replicates) of numbers of zero branches detected by adaptive phylogenetic LASSO for various regularization coefficients.}\label{fig:numzerobranch}
\end{center}
\end{figure}

\begin{figure}
\begin{center}
\includegraphics[width=\textwidth]{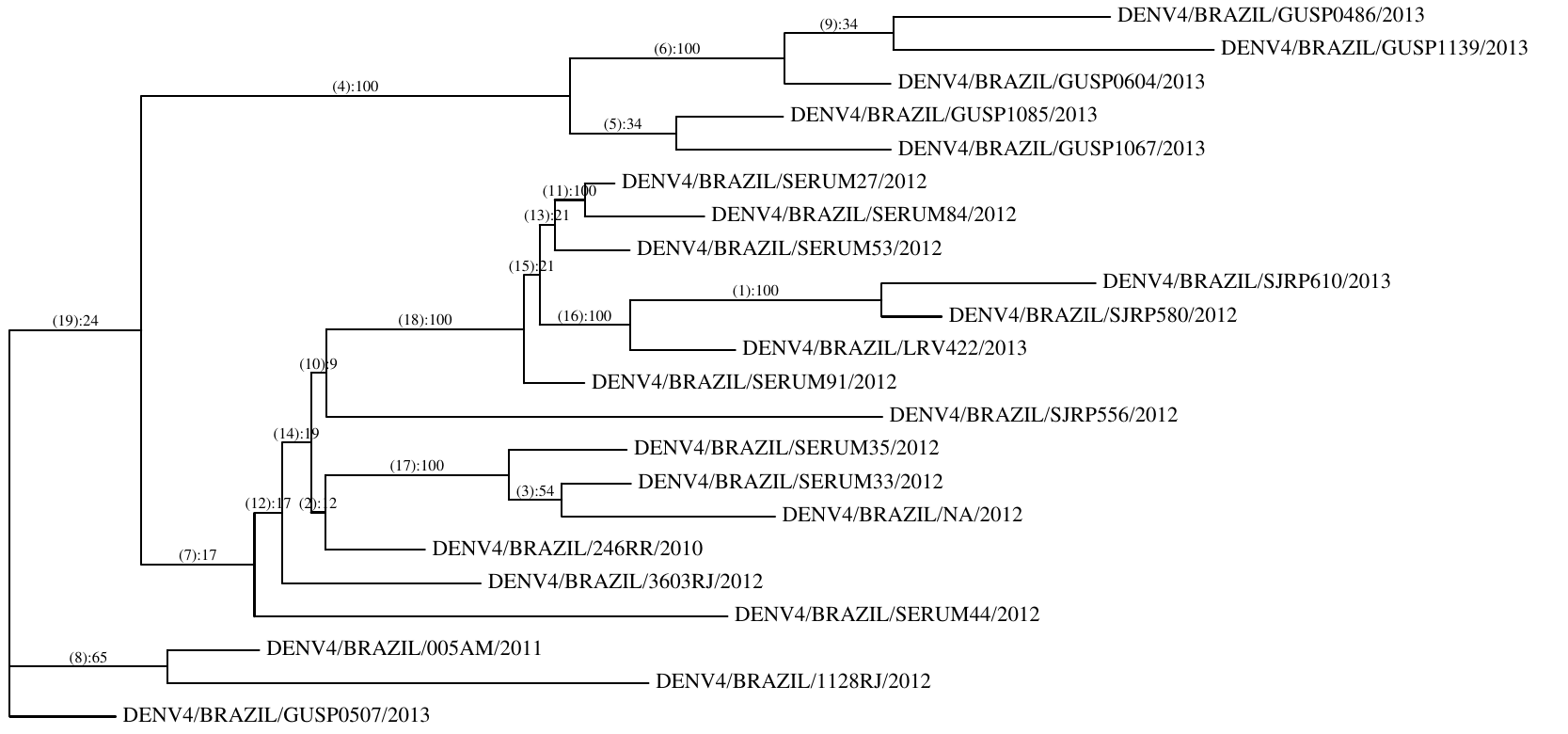}
\caption{Consensus tree resulting from the conventional MCMC Bayesian inference on the Brazil clade from DENV4.
% Consensus tree is obtained using p4, from a MrBayes run.
}\label{fig:consensus}
\end{center}
\end{figure}

Figure \ref{fig:numzerobranch} shows the distributions (across bootstrap replicates) of the numbers of zero branches detected by adaptive phylogenetic LASSO as a function of regularization coefficient $\lambda$.
We see lots of detected zero branches, indicating the existence of non-bifurcating topologies for this data set.
Furthermore, we find that even unpenalized ($\lambda = 0$) optimization recovers quite a few zero branches with our more thorough optimization allowing branch lengths to go all the way to zero.

Figure \ref{fig:consensus} shows the consensus tree obtained from the conventional MCMC samples.
Each interior edge has its index number $i$ and its support value (expressed as percentage) $s_i$ right above it: $\{(i):s_i\}$.
%If an edge has inferred length zero then it doesn't count in this bootstrap edge support.
%We see that many short edges have support below $50$, which indicates the strong preference for non-bifurcating topologies.
%We see that many short edges have relatively low support, which could indicate the existence of non-bifurcating topologies.
Following \citet{Lewis2005-ez}, we re-estimate the support values for all interior edges (splits) on this MCMC consensus tree using the aforementioned methods and summarize the results in Table \ref{tab:support}.
As one of the state-of-the-art approaches for identifying non-bifurcating topologies, rjMCMC is able to detect edges with exactly zero support (edges 2, 10, 13, 15).
Due to their sparsity-encouraging nature, ML+thresholding+bootstrap and ML+adaLASSO+bootstrap can identify these zero edges (compared to standard MCMC and ML+bootstrap) and the detected zero edges are largely consistent with rjMCMC.
Moreover, we also examine the support estimates of all splits observed in the standard MCMC samples, and find that adaptive phylogenetic LASSO tends to provide the closest estimates to rjMCMC (for zero-edge detection) among all the alternatives (Figure \ref{fig:wholesplit}).
Overall, we see that (multistep) adaptive phylogenetic LASSO is able to reveal non-bifurcating structures comparable to rjMCMC Bayesian approach when applied to maximum likelihood tree topologies, and is less likely to misidentify weakly supported edges in contrast to simple thresholding.

\setlength{\tabcolsep}{4pt}
\begin{table}
\ra{1.2}
\tiny
\begin{center}
\begin{tabularx}{1.0\textwidth}{@{}cccccccccccc@{}}
\toprule
\multirow{2}{*}{Edge}     &    ML   &  \multirow{2}{*}{MCMC}  &  rjMCMC  &\phantom{} &  \multicolumn{3}{c}{ML+thresholding+bootstrap} & \phantom{} &  \multicolumn{3}{c}{ML+adaLASSO+bootstrap} \\
\cmidrule{4-4} \cmidrule{6-8}\cmidrule{10-12}
                                      &   bootstrap  &       &  $C=1$    &&  $\kappa = 1e\text{-}06$  &  $\kappa = 5e\text{-}05$  &  $\kappa = 1e\text{-}04$  &&  $\lambda=150$  & $\lambda=300$ &  $\lambda=450$ \\
\midrule
  1 & 100 & 100 & 100 && 100 & 100 & 100 && 100 & 100 & 100 \\
  2 &   7 &  12 &   0 &&   0 &   0 &   0 &&   0 &   0 &   0 \\
  3 &  93 &  54 &   1 &&  82 &  65 &  58 &&  66 &  66 &  66 \\
  4 &  95 & 100 & 100 &&  95 &  95 &  95 &&  95 &  95 &  95 \\
  5 &  38 &  34 &   1 &&   0 &   0 &   0 &&   4 &   4 &   4 \\
  6 &  63 & 100 & 100 &&  61 &  61 &  61 &&  63 &  63 &  63 \\
  7 &  10 &  17 &  20 &&   8 &   8 &   7 &&   6 &   6 &   6 \\
  8 &  52 &  65 &  31 &&  45 &  44 &  44 &&  44 &  44 &  44 \\
  9 &  11 &  34 &   1 &&   0 &   0 &   0 &&   2 &   2 &   2 \\
 10 &  10 &   9 &   0 &&   0 &   0 &   0 &&   0 &   0 &   0 \\
 11 &  61 & 100 &  66 &&  57 &  56 &  22 &&  61 &  61 &  61 \\
 12 &  13 &  17 &   3 &&   8 &   2 &   0 &&   3 &   3 &   3 \\
 13 &   9 &  21 &   0 &&   0 &   0 &   0 &&   0 &   0 &   0 \\
 14 &  23 &  19 &   3 &&   9 &   2 &   0 &&   4 &   3 &   4 \\
 15 &  13 &  21 &   0 &&   0 &   0 &   0 &&   0 &   0 &   0 \\
 16 &  99 & 100 & 100 &&  99 &  99 &  96 &&  99 &  99 &  99 \\
 17 &  94 & 100 & 100 &&  94 &  94 &  94 &&  94 &  94 &  94 \\
 18 &  88 & 100 & 100 &&  88 &  88 &  88 &&  88 &  88 &  88 \\
 19 &   6 &  24 &  24 &&   4 &   4 &   3 &&   3 &   3 &   3 \\
\bottomrule
\end{tabularx}
\caption{Comparison on the support values obtained from different methods on DENV4 Brazil clade data set. All analysis used the Jukes-Cantor model.}\label{tab:support}
\end{center}
\end{table}

\begin{figure}
\begin{center}
\includegraphics[width=\textwidth]{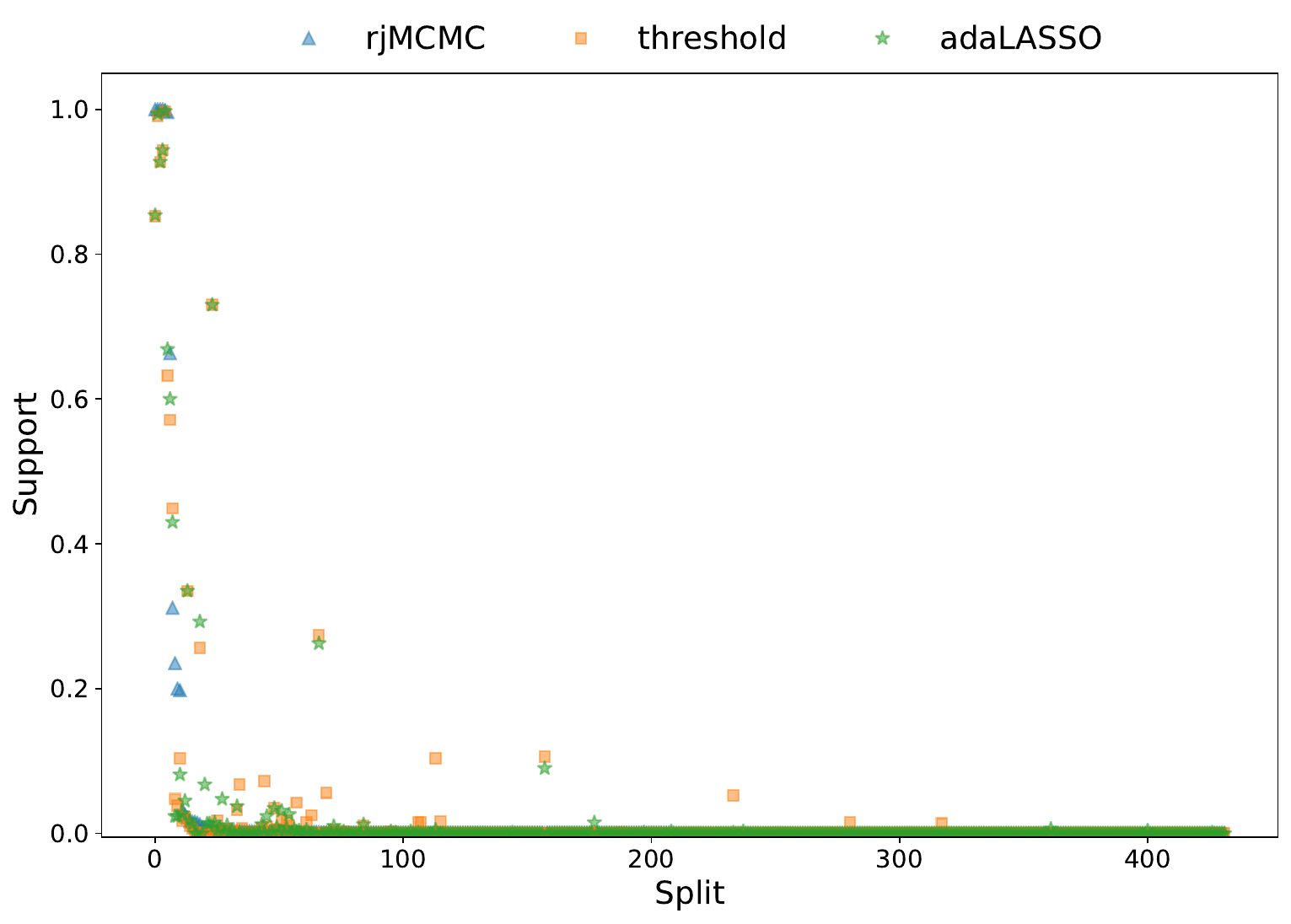}
\caption{A comparison of different methods on the support estimates of all splits observed in the MCMC samples. %We restrict the $y$ axis to better visualize the less-likely splits.
    Splits are sorted by their support under rjMCMC.
}\label{fig:wholesplit}
\end{center}
\end{figure}

\section{Conclusion}
We study $\ell_1$-penalized maximum likelihood approaches for phylogenetic inference, with the goal of recovering non-bifurcating tree topologies.
We prove that these regularized maximum likelihood estimators are asymptotically consistent under mild conditions.
Furthermore, we show that the (multistep) adaptive phylogenetic LASSO is topologically consistent and therefore is able to detect non-bifurcating tree topologies that may contain polytomies and sampled ancestors.
We present an efficient algorithm for solving the corresponding optimization problem, which is inherently more difficult than standard $\ell_1$-penalized problems with regular cost functions.
The algorithm is based on recent developments on proximal gradient descent methods and their various acceleration techniques \citep{Beck2009-ty,donoghue13}.

Our method is closest in spirit to rjMCMC, which is a rigorous means of inferring the posterior distribution of potentially multifurcating topologies, and thus we have limited our performance comparisons to this method.
However, there is precedent for using a hypothesis-testing framework to test if parts of the tree should be multifurcating.
\citet{Jackman1999-lv} use a parametric bootstrap test and a randomization test built on maximum parsimony to evaluate the strength of support for a rapid-evolution scenario.
\citet{Walsh1999-ow} determine the number of base pairs required to resolve a given phylogenetic divergence with a prior hypothesis about the amount of time during which this divergence could have occurred.
Also, one might wish to use the SOWH test \citep{sowhOriginal,Goldman2000-qw,Susko2014-no} to find an appropriate threshold by increasing the threshold progressively until the SOWH test detects a significant difference between the ML tree and the thresholded tree.
However, we have shown that thresholding is less effective than our method across a range of threshold values (Figure~\ref{fig:consistency}).
Furthermore, existing software implementations of SOWH cannot test multifurcating topologies because inference under a multifurcation constraint is not supported in current tree inference packages.

We have done a wide range of experiments to demonstrate the efficiency and effectiveness of our method.
We show in a synthetic study that although the (non-adaptive) phylogenetic LASSO has difficulty finding zero-length branches, the adaptive phylogenetic LASSO provides significant improvement on sparsity recovery which validates its theoretical properties.
Although we assume a fixed tree topology for deriving statistical consistency, our method can be used to discover non-bifurcating tree topologies in real data problems when combined with traditional maximum likelihood phylogenetic inference methods.
Our experiments have shown that the adaptive phylogenetic LASSO performs comparably with other MCMC based sparsity encouraging procedures (rjMCMC) in terms of sparsity recovery while being computationally more efficient as an optimization approach.
We also compare our method to a heuristic simple thresholding approach and find that regularization permits more consistent performance.
Finally, we show that compared to rjMCMC, the adaptive phylogenetic LASSO is more likely to detect short branches while identifying zero branches with high accuracy.
It is worth mentioning that while lots of sparsity can be detected by maximizing the likelihood with non-negative constraints, the adaptive phylogenetic LASSO can be advantageous when there exist challenging zero-branches in the tree topologies with high likelihoods.
Our results offer new insights into non-bifurcating phylogenetic inference methods and support the use of $\ell_1$ penalty in statistical modeling with more general settings.

We leave some questions to future work.
For the theory, the rate with which we can allow the number of leaves to go to infinity in terms of the sequence length is not yet known.
We have also not explored the extent to which the optimal penalized tree is a contraction of the ML unpenalized tree.
Also, although we have laid the algorithmic foundation for efficient penalized inference, there is further work to be done to make a streamlined implementation that is integrated with existing phylogenetic inference packages.

\section{Acknowledgements}
The authors would like to thank Vladimir Minin and Noah Simon for helpful discussions, and Sidney Bell for helping with the Dengue sequence data.
This work supported by National Institutes of Health grants R01-GM113246, R01-AI120961, U19-AI117891, and U54-GM111274 as well as National Science Foundation grant CISE-1564137.
The research of Frederick Matsen was supported in part by a Faculty Scholar grant from the Howard Hughes Medical Institute and the Simons Foundation.

\bibliographystyle{chicago}
\bibliography{l1-reg}

\beginsupplement

\section{Appendix}
\subsection{Lemmas}

Here we perform further theoretical development to establish the main theorems.
We remind the reader that we will continue to assume Assumptions~\ref{assump:iden} and~\ref{assump:leafedges}.
The following lemma allows gives a lower bound on the fraction of sites with state assignments in a given set.
It will prove useful to obtain an upper bound on the likelihood.

\begin{Lemma}
For any non-empty set $A$ of single-site state assignments to the leaves, we define
\[
k_A = |\{i: \mb{Y}^i \in A\}|
\]
There exist $c_3>0, c_4(\delta, n)>0$ such that for all $k$, we have
\[
\frac{k_A}{k} \ge c_3 - \frac{c_4}{\sqrt{k}} \h \forall A \ne \emptyset
\]
with probability at least $1-\delta$.
\label{lem:frequency}
\end{Lemma}

\begin{proof}[Proof of Lemma $\ref{lem:frequency}$]
Since the tree distance between any pairs of leaves of the true tree is strictly positive, there exists $c_3>0$ such that $P_{q^*}(\psi) \ge c_3$ for all state assignments $\psi$.

Using Hoeffding's inequality, for any state assignment $\psi$, we have
\[
\mathbb{P}\left[\left|\frac{k_{\{\psi\}}}{k} - P_{q^*}(\psi)\right| \ge t\right] \le 2 e^{-2kt^2}.
\]
We deduce that
\[
\mathbb{P}\left[\exists \psi \text{ such that} \left|\frac{k_{\{\psi\}}}{k} - P_{q^*}(\psi)\right| \ge t\right] \le 2 e^{-2kt^2} \cdot 4^N.
\]
For any given $\delta>0$, by choosing
\[
c_4(\delta, N) =\sqrt{\frac{\log(1/\delta) + (2N+1) \log 2}{2}}
\]
and $t = c_4(\delta, N) / \sqrt{k}$ we have
\[
\left|\frac{k_{\{\psi\}}}{k} - P_{q^*}(\psi)\right| \le   \frac{c_4(\delta, N)}{\sqrt{k}} \h \forall \psi
\]
with probability at least $1-\delta$.
This proves the Lemma.
\end{proof}

\begin{Lemma}[Generalization bound]
There exists a constant $C(\delta, n, Q, \eta, g_0, \mu) > 0$ such that for any $k \ge 3$, $\delta >0$, we have:
\[
\left |\frac{1}{k}\ell_k(q) - \phi(q)\right | \le C  \left(\frac{\log k}{k}\right)^{1/2} \h \forall q \in \mathcal{T}(\mu)
\]
with probability greater than $1-\delta$.
\label{lem:unifboundg0}
\end{Lemma}

\begin{proof}
Note that for $q \in \mathcal{T}(\mu)$, $0\ge \log P_q(\psi) \ge -\mu$ for all state assignments $\psi$.
By Hoeffding's inequality,
\[
\mathbb{P}\left[\left |\frac{1}{k}\ell_k(q) - \phi(q)\right |  \ge y/2 \right] \le 2 \exp\left(-\frac{y^2 k}{2\mu^2}\right).
\]
For each $q \in \mathcal{T}(\mu)$, $k > 0$, and $y > 0$, define the events
\begin{equation*}
A({q, k, y}) = \left\{\left |\frac{1}{k}\ell_k(q) - \phi(q)\right |> y/2  \right \}
\end{equation*}
and
\[
B({q, k, y}) = \left\{ \exists q' \in \mathcal{T}(\mu) ~\text{such that}~ \|q'-q\|_2 \le \frac{y}{4c_2}~~\text{and}~~ \left |\frac{1}{k}\ell_k(q) - \phi(q)\right |> y  \right  \}
\]
then $B({q, k, y}) \subset A({q, k, y})$ by the triangle inequality, \eqref{eqn:union1}, and \eqref{eqn:union2}.
Let
\[
y = \sqrt{\frac{C \log k}{k}}
\]
Since $\mathcal{T}(\mu)$ is a subset of $\mathbb{R}^{2N-3}$, there exist $C_{2N-3} \ge 1$ and a finite set $\mathcal{H} \subset \mathcal{T}(\mu)$ such that
\[
\mathcal{T}(\mu) \subset \bigcup_{q \in \mathcal{H}}{V(q, \epsilon)} \h \text{and}\h  |\mathcal{H}| \le C_{2N-3} /\epsilon^{2N-3}
\]
where $\epsilon=y/(4c_2)$, $V(q, \epsilon)$ denotes the open ball centered at $q$ with radius $\epsilon$, and $|\mathcal{H}|$ denotes the cardinality of $\mathcal{H}$.
By a simple union bound, we have
\[
\mathbb{P}\left[ \exists q \in \mathcal{H}: \left |\frac{1}{k}\ell_k(q) - \phi(q)\right |> y/2\right]  \le 2 \exp\left(-\frac{y^2  k }{2\mu^2}\right) C_{2N-3} \Big/ \epsilon^{2N-3}.
\]
Using the fact that $B({q, k, y}) \subset A({q, k, y})$ for all $q \in \mathcal{H}$, we deduce
\[
\mathbb{P}\left[ \exists q  \in \mathcal{T}(\mu) : \left |\frac{1}{k}\ell_k(q) - \phi(q)\right |> y\right]  \le 2 \exp\left(-\frac{y^2  k }{2\mu^2}\right) C_{2N-3} \Big/ \epsilon^{2N-3}.
\]
To complete the proof, we need to chose $C$ in such a way that
\[
C_{2N-3} \left(\frac{4\sqrt{k}g_0 c_2}{\sqrt{C\log k}}\right)^{2N-3}\times 2\exp\left(-\frac{ C \log k}{2 \mu^2}\right) ~ \le~ \delta.
\]
Since $k \ge 3$ and $C \ge 1$, the inequality is valid if
\[
C_{2N-3} \left(4g_0 c_2\right)^{2N-3}\times 2k^{\frac{2N-3}{2}- \frac{C}{2 \mu^2}} ~ \le~ \delta
\]
and can be obtained if
\[
\frac{2N-3}{2}- \frac{C}{2 \mu^2} <0, \h \text{and} \h C_{2N-3} \left(4g_0 c_2\right)^{2N-3}\times 2\cdot 3^{\frac{2N-3}{2}- \frac{C}{2 \mu^2}} ~ \le~ \delta.
\]
In other words, we need to choose $C$ such that
\[
C \ge  2 \mu^2 \left( \log (1/\delta) + \log C_{2N-3}  + (2N-3) \log(4\sqrt{3}g_0 c_2)\right).
\]
This completes the proof.
\end{proof}

\subsection{Proofs of main theorems}

\begin{proof}[Proof of Theorem $\ref{theo:base}$]
By definition of the estimator, we have
\[
- \frac{1}{k}\ell_k( q^{k, R_k}) +  \lambda_k  R_k(q^{k, R_k}) \le - \frac{1}{k}\ell_k(q^*) +  \lambda_k R_k(q^*)
\]
which is equivalent to $U_k( q^{k, R_k}) \le \lambda_k R_k(q^*) - \lambda_k R_k(q^{k, R_k})$.

We have $q^{k, R_k} \in \mathcal{T}(\mu)$ with probability at least $1-2\delta$ from Lemma $\ref{lem:Tmu}$ for $k$ sufficiently large.
Therefore by Lemma~\ref{lem:unifbound-onesided},
\[
\mathbb{E}[U_k(q^{k, R_k})] \le \frac{1}{k} \h \text{or} \h \frac{1}{2}\mathbb{E}[U_k(q^{k, R_k})]  \le U_k(q^{k, R_k}) + \frac{C\log k}{k^{2/\beta}},
\]
with probability at least $1-3\delta$.
The second case implies that
\begin{align*}
\frac{c_1^{\beta}}{2} \|q^{k, R_k}-q^*\|_2^{\beta} & \le \frac{1}{2} \mathbb{E}[U_k(q^{k, R_k})] \\
& \le \lambda_k R_k(q^*) - \lambda_k  R_k(q^{k, R_k}) + \frac{C \log k}{k^{2/\beta}}
\le \frac{C \log k}{k^{2/\beta}} + \lambda_k R_k(q^*)
\end{align*}
while for the first case, we have
\[
\frac{c_1^{\beta}}{2}  \|q^{k, R_k}-q^*\|_2^{\beta} \le \mathbb{E}[U_k(q^{k, R_k})] \le \frac{1}{k} \le  \frac{C \log k}{k^{2/\beta}} + \lambda_k R_k(q^*)
\]
since $\beta \ge 2$ and $C \ge 1$.
This demonstrates \eqref{eqn:baseOne}.

If the additional assumption \eqref{eqn:boost} is satisfied, we also have
\[
\|q^{k, R_k}-q^*\|_2^{\beta} \le  \frac{C' \log k}{k^{2/\beta}} + C_3 \lambda_k\|q^{k, R_k}-q^*\|_2.
\]
Using Lemma $\ref{lem:boost}$ with
\[
\nu = 1/\beta, \h x=\|q^{k, R_k}-q^*\|_2^{\beta}, \h a =C_3 \lambda_k \h  \text{and} \h b= \frac{C' \log k}{k^{2/\beta}},
\]
we obtain
\[
x \le C_1 a^{1/(1-\nu)} + C_2 b,
\]
which implies
\[
\|q^{k, R_k}-q^*\|_2^{\beta} \le C'(\delta, C_3) \left (\frac{\log k}{ k^{2/\beta}}+\lambda_k^{\beta/(\beta-1)}  \right).
\]
This completes the proof.
\end{proof}

\begin{proof}[Proof of Theorem $\ref{theo:adapt}$]
We first note that by Theorem $\ref{theo:base}$, the estimator $q^{k, R_k}$  is consistent, which guarantees $\lim_{k \to \infty}{q^{k, R_k}} = q^*$ almost surely.
Thus
\[
\lim_{k \to \infty}{S_k(q^*)}= \lim_{k \to \infty} {\sum_{q^*_i \ne 0}{ (q^*_i)^{1-\gamma}}} < \infty.
\]
The hypotheses of this theorem imply that $\lambda_k \to 0$ and thus by Theorem $\ref{theo:base}$, we also deduce that $q^{k, S_k}$ is also a consistent estimator.
This validates (i).

To establish topological consistency under (ii), we divide the proof into two steps.

As the first step, we prove that  $\lim_k{\mathbb{P}(\mathcal{A}(q^*) \subset \mathcal{A}(q^{k, S_k})) }=1.$
If $q^*_{i_0} =0$ for some $i_0$, then from Theorem $\ref{theo:base}$, we have
\[
q^{k, R_k}_{i_0} \le C'(\delta) \left (\frac{\log k}{k^{2/\beta}}+\lambda_k^{\beta/(\beta-1)} \right)^{1/\beta} \h \forall k
\]
with probability at least $1-\delta$.
By the definition of $w_{k,i_0}$, we have
\begin{eqnarray*}
\lim_{k \to \infty}{~\alpha_k w_{k,i_0} } &\ge& \lim_{k \to \infty}{~\alpha_k (C'(\delta))^{-\gamma}  \left (\frac{\log k}{k^{2/\beta}}+\lambda_k^{\beta/(\beta-1)} \right)^{-\gamma/\beta}} \\
&=& (C'(\delta))^{-\gamma} ~ \lim_{k \to \infty}{~  \left (\frac{\log k}{\alpha_k^{\beta/\gamma}k^{2/\beta}}+\alpha_k^{-\beta/\gamma} \lambda_k^{\beta/(\beta-1)}   \right)^{-\gamma/\beta}} \\
\end{eqnarray*}
which goes to infinity since by the hypotheses of the Theorem
\[
\alpha_k^{\beta/\gamma} \succ  \frac{\log k}{k^{2/\beta}} \h \text{and} \h \alpha_k^{\beta / \gamma} \succ \lambda_k^{\beta/(\beta-1)}.
\]
Since $\delta>0$ is arbitrary, we deduce that $\lim_{k \to \infty}{~\alpha_k w_{k,i_0} } = \infty$ with probability one.

Now for any branch length vector $q$, we define $f(q)$ as the vector obtained from $q$ by setting the $i_0$ component of $q$ to 0.
By definition of the estimator $q^{k, S_k}$, we have
\[
- \frac{1}{k}\ell_k( q^{k, S_k}) +  \alpha_k  \sum_{i}{w_{k,i} \, q^{k, S_k}_i} \le - \frac{1}{k}\ell_k(f(q^{k, S_k})) +  \alpha_k  \sum_{i}{w_{k,i} [f(q^{k, S_k})]_i}
\]
or equivalently
\[
 \alpha_k w_{k,i_0}q^{k, S_k}_{i_0} \le  \frac{1}{k}\ell_k(q^{k, S_k}) - \frac{1}{k}\ell_k(f(q^{k, S_k})).
\]
Lemma $\ref{lem:star}$ establishes that there exist, $\mu^*>0$ and a neighborhood $V$ of $q^*$ in $\mathcal{T}$ such that $V \subset \mathcal{T}(\mu^*)$.
Since the estimator $q^{k, S_k}$ is consistent and $q^*_{i_0} =0$, we can assume that both $q^{k, S_k}$ and $f(q^{k, S_k})$ belong to $\mathcal{T}(\mu^*)$ with $k$ large enough.
Thus, from Lemma $\ref{lem:locallips}$, we have
\[
\left|\frac{1}{k}\ell_k(q^{k, S_k}) - \frac{1}{k}\ell_k(f(q^{k, S_k})) \right| \le c_2 \|q^{k, S_k} - f(q^{k, S_k})\|_2 = c_2 q^{k, S_k}_{i_0}.
\]
If $q^{k, S_k}_{i_0} >0$, we deduce that $\alpha_k w_{k,i_0}$ is bounded from above by $c_2$, which is a contradiction.
This implies that $q^{k, S_k}_{i_0} =0$, and we conclude that
\[
\lim_k{\mathbb{P}(\mathcal{A}(q^*) \subset \mathcal{A}(q^{k, S_k})) }=1.
\]

As the second step, we prove that  $\lim_k{\mathbb{P}(\mathcal{A}(q^{k, S_k}) \subset \mathcal{A}(q^*)) }=1$.
Indeed, the consistency of $q^{k, S_k}$ guarantees that
\[
\lim_{k \to \infty}{q^{k, S_k}} = q^*
\]
almost surely.
Therefore, if $q^*_{i_0} > 0$ for some $i_0$, then $q^{k, S_k}_{i_0} > 0$ for $k$ large enough.
In other words, we have $\lim_k{\mathbb{P}(\mathcal{A}(q^{k, S_k}) \subset \mathcal{A}(q^*)) }=1$.

Combing step 1 and step 2, we deduce that the adaptive estimator is topologically consistent.
\end{proof}

\begin{proof}[Proof of Lemma $\ref{lem:lips}$]
Since $q^{k, S_k}$ is topologically consistent and $q^{k, R_k}$ is consistent, we have
\[
\mathcal{A}(q^{k, S_k}) =\mathcal{A}(q^*) \h \text{and} \h q^{k, R_k}_i \ge q_i^*/2 \h  \forall i \not \in \mathcal{A}(q^*)
\]
with probability one for sufficiently large $k$.
Defining $b = \min_{i \not \in \mathcal{A}(q^*)}{q_i^*}$, we have
\[
|S_k(q^{k, S_k}) - S_k(q^*) |= \left|\sum_{q_i^* \ne 0}{w_{k,i} (q^{k, S_k}_{i}-q_i^*)}\right| \le \sqrt{2N-3}\,(b/2)^{-\gamma}\,\|q^{k, S_k}-q^*\|_2
\]
via Cauchy-Schwarz which completes the proof.
\end{proof}

\begin{proof}[Proof of Theorem $\ref{theo:convergence}$]
We note that for the LASSO estimator, $R_k^{[0]}(q^*) = \sum_{i}{q^*_i}$ is uniformly bounded from above.
Hence, the LASSO estimator is consistent.
We can then use this as the base case to prove, by induction, that adaptive LASSO and the multiple-step LASSO are consistent via Theorem $\ref{theo:adapt}$ (part (i)).
Moreover, $R_k^{[0]}$ is uniformly Lipschitz and satisfies \eqref{eqn:boost}, so using part (ii) of Theorem $\ref{theo:adapt}$, we deduce that adaptive LASSO (i.e., the estimator with penalty function $R_k^{[1]}$) is topologically consistent.

We will prove that the multiple-step LASSOs are topologically consistent by induction.
Assume that $q^{k, R^{[m]}_k}$ is topologically consistent, and that $q^{k, R^{[m-1]}_k}$ is consistent.
From Lemma~$\ref{lem:lips}$, we deduce that there exists $C>0$ independent of $k$ such that
\begin{equation}
\left |R^{[m]}_k\left(q^{k, R^{[m]}_k} \right) - R^{[m]}_k(q^*) \right| \le C \left\|q^{k, R_k^{[m]}}-q^* \right\|_2 \h \forall k.
\label{eqn:induction}
\end{equation}
This enables us to use part (ii) of Theorem $\ref{theo:adapt}$ to conclude that $q^{k, R^{[m+1]}_k}$ is topologically consistent.
This inductive argument proves part (i) of the Theorem.
We can now use \eqref{eqn:induction} and Theorem~$\ref{theo:base}$ to derive the convergence rate of the estimators.
\end{proof}

\subsection{Technical proofs}

\REexist*
\begin{proof}[Proof of Lemma $\ref{lem:exist}$]
Let $\{q^n\}$ be a sequence such that
\[
Z_{\lambda, \mb{Y}^k}(q^n) \to \nu := \inf_{q} Z_{\lambda, \mb{Y}^k} (q).
\]
We note that since $\ell_k(q^*) \ne -\infty$ and $R_k$ is continuous on the compact set $\mathcal{T}$, $\nu$ is finite.
Since $\mathcal{T}$ is compact, we deduce that a subsequence $\{q^m\}$ converges to some $q^0\in \mathcal{T}$.
Since the log likelihood (defined on $\mathcal{T}$ with values in the extended real line $[-\infty, 0]$) and the penalty $R_k$ are continuous, we deduce that $q^0$ is a minimizer of $Z_{\lambda, \mb{Y}^k}$.
\end{proof}

\RElocallips*
\begin{proof}[Proof of Lemma $\ref{lem:locallips}$]
Using the same arguments as in the proof of Lemma~4.2 of \citet{dinh2018consistency}, we have
\[
\left| \frac{\partial P_q(\psi)}{ \partial q_i} \right| \le \varsigma 4^{n}
\]
for any state assignment $\psi$ where $\varsigma$ is the element of largest magnitude in the rate matrix $Q$.
By the Mean Value Theorem, we have
\begin{equation*}
|\log P_{q}(\psi) - \log P_{q'}(\psi)| \le c_2 \sqrt{2N-3}\|q-q'\|_2 \h \forall q,q', \psi
\end{equation*}
% Erick note to self: we need the e^{-mu} because of chain rule (we are bounding log P).
where $c_2 := \varsigma 4^n/e^{-\mu}$, and $\|\cdot\|_2$ is the $\ell_2$-distance in $\mbb{R}^{2N-3}$.
This implies both \eqref{eqn:union1} and \eqref{eqn:union2}.
\end{proof}

\REunifboundOnesided*
\begin{proof}[Proof of Lemma $\ref{lem:unifbound-onesided}$]
The difference of average likelihoods $U_{k}(q)$ is bounded by Lemma~\ref{lem:locallips} and the boundedness assumption on $\mathcal{T}$, thus by Hoeffding's inequality
\[
\mathbb{P}\left[U_{k}(q) - \mathbb{E}\left [U_{k}(q) \right]    \le -y \right] \le \exp\left(-\frac{2y^2 k}{c_2^2\|q-q^*\|^2}\right).
\]
By choosing $y =  \frac{1}{2} \mathbb{E}\left [U_{k}(q)\right] + t/2$, we have $y^2 \ge t\mathbb{E}\left [U_{k}(q)\right]$.
For any $q \in G_k$, we deduce using \eqref{eq:varExp} (and the fact that $\beta \ge 2$) that
\[
\mathbb{P}\left[ U_{k}(q)  \le \frac{1}{2}  \mathbb{E} \left [U_{k}(q) \right] -t/2 \right]
\le \exp\left(-\frac{2c_1^2 t k \mathbb{E}[U_{k}(q)]}{c_2^2 \mathbb{E}[U_{k}(q)]^{2/\beta}}\right)
\le \exp\left(-\frac{2c_1^2 t  k^{2/\beta} }{c_2^2}\right).
\]
For each $q \in G_k$, define the events
\begin{equation*}
A({q, k, t}) = \left\{U_{k}(q) -\frac{1}{2} \mathbb{E}\left [U_{k}(q)\right] \le - t/2  \right \}
\end{equation*}
and
\[
B({q, k, t}) = \left\{ \exists q' \in G_k ~\text{such that}~ \|q'-q\|_2 \le \frac{t}{4c_2}~~\text{and}~~ U_{k}(q') - \frac{1}{2} \mathbb{E}\left [U_{k}(q')\right] \le -t   \right  \}
\]
then $B({q, k, t}) \subset A({q, k, t})$ by the triangle inequality, \eqref{eqn:union1}, and \eqref{eqn:union2}.
Let
\[
t = \frac{C\log k}{k^{2/\beta}}.
\]
To obtain a union bound and complete the proof, we need to chose $C$ in such a way that
\[
C_{2N-3} \left(\frac{4k^{2/\beta}g_0 c_2}{C \log k}\right)^{2N-3}\times 2\exp\left(-\frac{2 c_1^2 C \log k}{c_2^2 }\right) ~ \le~ \delta
\]
where $C_{2N-3}$ is defined as in the proof of Lemma $\ref{lem:unifboundg0}$.
This can be done by choosing
\[
C \ge\frac{ 4 \beta c_2^2 }{9c_1^2}\left( \log (1/\delta) + \log C_{2N-3}  + (2N-3) \log(4 \cdot 3^{2/\beta}g_0 c_2)\right).
\]
\end{proof}

\REstar*
\begin{proof}[Proof of Lemma $\ref{lem:star}$]
Let
\[
\mu^* = - 2 \min_{\psi}{\log P_{q^*}(\psi)}
\]
then we have $\log P_{q^*}(\psi) > - \mu^*$ for all state assignments $\psi$.

For a fixed value of $\psi$, $\log P_{q}(\psi)$ is a continuous function of $q$ around $q^*$.
Hence, there exists an neighborhood $V_{\psi}$ of $q^*$ such that $V_{\psi}$ is open in $\mathcal{T}$ and $\log P_{q}(\psi) > - \mu^*$.
Let $V= \cap_{\psi}{V_{\psi}}$.
Because the set of all possible labels $\psi$ of the leaves is finite, $V$ is open in $\mathcal{T}$ and
\[
\log P_{q}(\psi) > - \mu^* \h \forall \psi, \forall q \in V.
\]
In other words, we have $V \subset \mathcal{T}(\mu^*)$.
\end{proof}

\RETmu*
\begin{proof}[Proof of Lemma $\ref{lem:Tmu}$]
We first assume that $\mu > \mu^*$, where $\mu^*$ is defined in Lemma~$\ref{lem:star}$.
Thus, we have $q^* \in \mathcal{T}(\mu^*) \subset \mathcal{T}(\mu)$.
By definition, we have
\[
- \frac{1}{k}\ell_k( q^{k, R_k}) +  \lambda_k  R_k(q^{k, R_k}) \le - \frac{1}{k}\ell_k(q^*) +  \lambda_k R_k(q^*)
\]
which implies via Lemma $\ref{lem:unifboundg0}$ that
\begin{equation}
\phi(q^*)  - C(\delta) \frac{\log k}{\sqrt{k}}+  \lambda_k  R_k(q^{k, R_k})  -  \lambda_k R_k(q^*)\le  \frac{1}{k}\ell_k( q^{k, R_k})
\label{eq:boundphi}
\end{equation}
with probability at least $1-\delta$.

Let $c_3$ and $c_4(\delta,N)$ be as in Lemma~$\ref{lem:frequency}$, and assume that $k$ is large enough such that
\begin{equation}
c_3 - c_4(\delta,N) \frac{\log k}{\sqrt{k}} >0.
\label{eq:boundk}
\end{equation}
Denoting the upper bound of $\{\lambda_k R_k(q^*)\}$ by $U$, we define
\[
\mu= \max \left \{ - 2\left(c_3 - c_4(\delta,N) \frac{\log k}{\sqrt{k}} \right)^{-1} \left( \phi(q^*)  - C(\delta) \frac{\log k}{\sqrt{k}} -  U \right), \mu^* \right \}.
\]
If we assume that $q^{k, R_k} \not \in \mathcal{T}(\mu)$, then the set $I= \{\psi: \log P_{q^{k, R_k}}(\psi) \le - \mu \}$ is non-empty.
Using Lemma~$\ref{lem:frequency}$, we have
\begin{equation}
\frac{1}{k}\ell_k( q^{k, R_k}) \le \frac{1}{k} \sum_{Y_i \in I}{\log P_{q^{k, R_k}}(Y_i)} \le -\mu \cdot \frac{k_{I}}{k} \le -\mu \cdot \left(c_3 - c_4(\delta) \frac{\log k}{\sqrt{k}} \right)
\label{eq:boundfrequency}
\end{equation}
with probability at least $1-\delta$.

Combining equations $\eqref{eq:boundphi}$ and $\eqref{eq:boundfrequency}$, and using the fact that  $\{\lambda_k R_k(q^*)\}$ is bounded by $U$, we obtain
\[
\phi(q^*)  - C(\delta) \frac{\log k}{\sqrt{k}} -  U \le -\mu \cdot \left(c_3 - c_4(\delta,N) \frac{\log k}{\sqrt{k}} \right).
\]
% Erick note to self: this equation reduces to the LHS of the equation just previous <= 2 times it, which is a contradiction because they are both negative.
This contradicts the choice of $\mu$ for $k$ large enough such that $\eqref{eq:boundk}$ holds.
% Erick note to self: We have to have each of these events happen, each of which with probability 1-delta, so the probability that both of them happen ((1-delta)^2) is greater than 1-2delta.

We deduce that $q^{k, R_k} \in \mathcal{T}(\mu)$ with probability at least $1-2\delta$.
\end{proof}

\subsection{More experimental results}

Here we present additional experimental results for the case of $\gamma > 1$.

\begin{figure}[ht]
\begin{center}
\includegraphics[width=0.45\textwidth]{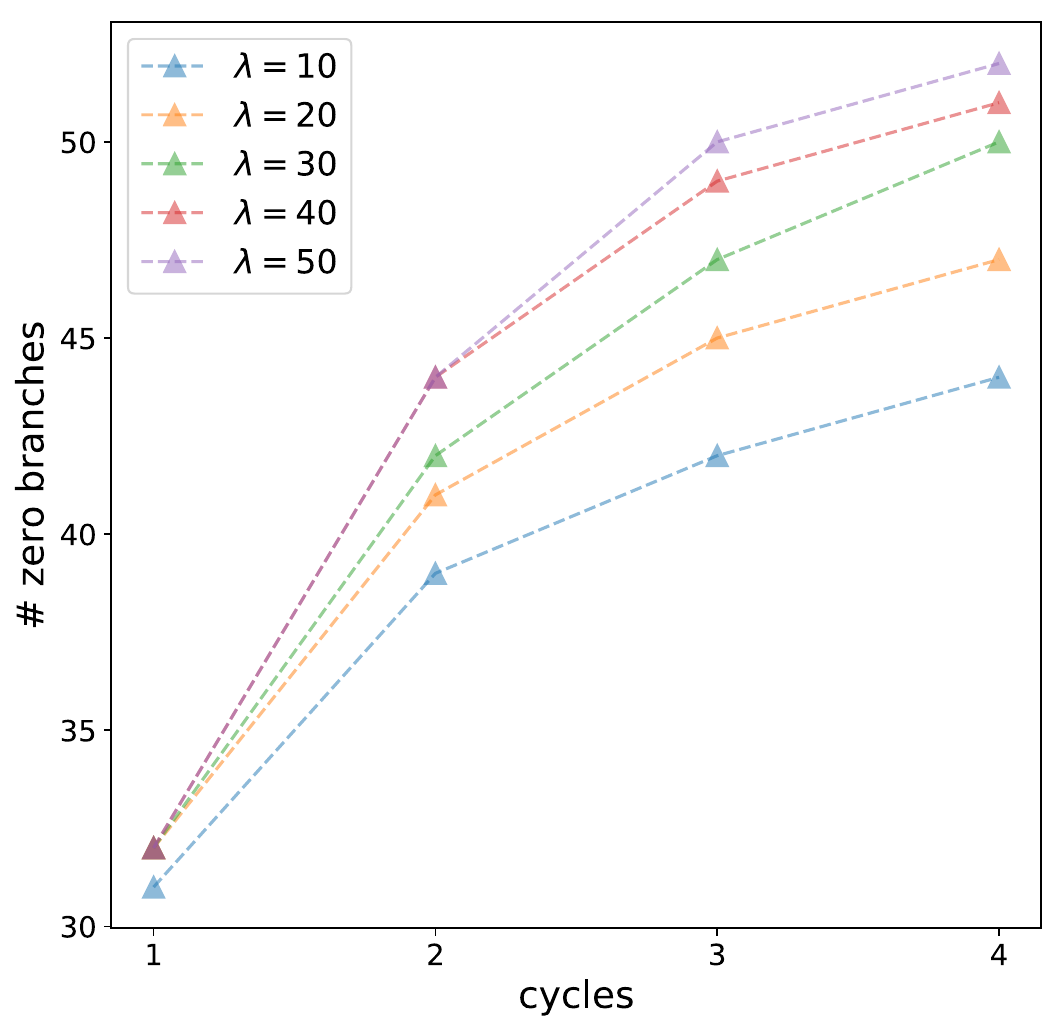}\hspace{20pt}
\includegraphics[width=0.45\textwidth]{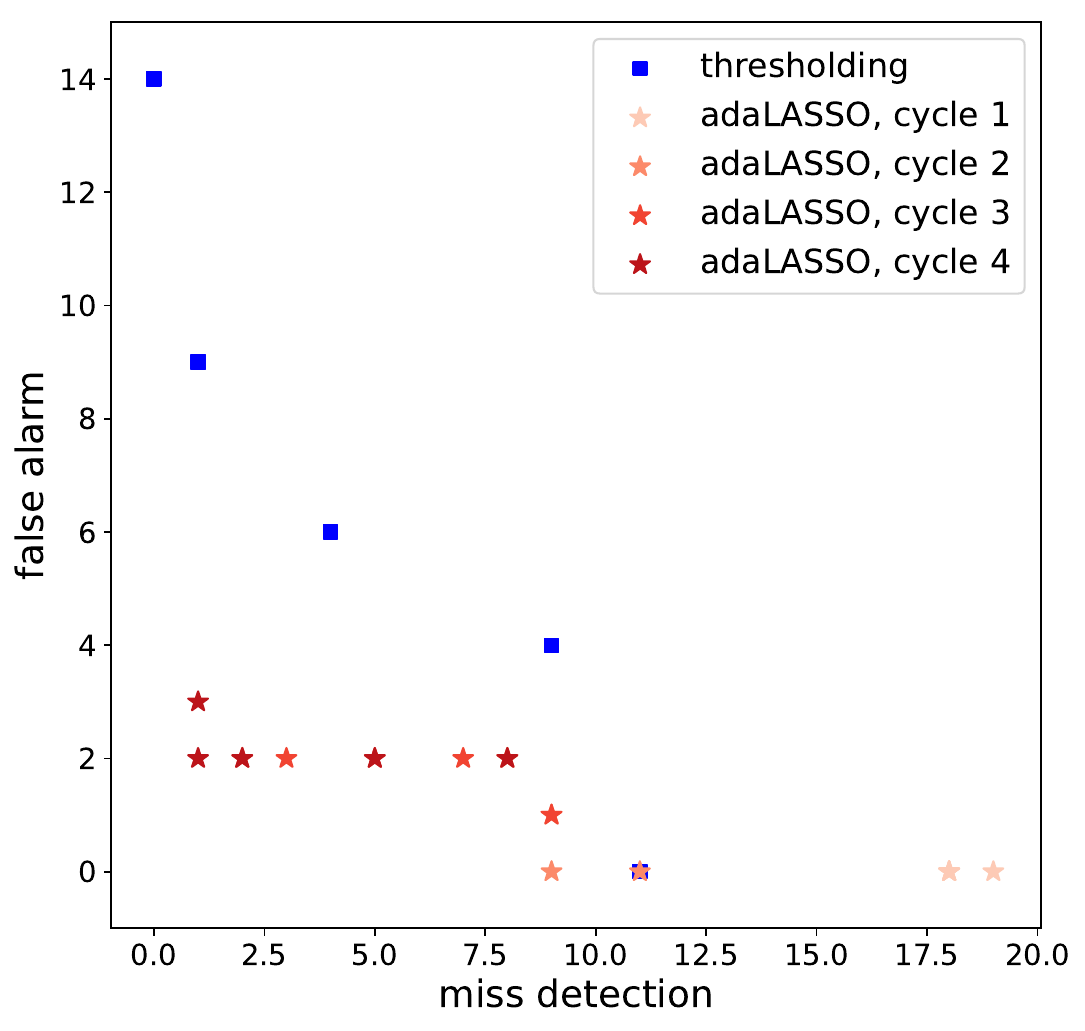}
\end{center}
\caption{Topological consistency comparison of different phylogenetic LASSO procedures on simulation 2. $\gamma=1.01$.}\label{fig:consistencygamma101}
\end{figure}

\begin{figure}[ht]
\begin{center}
\includegraphics[width=0.45\textwidth]{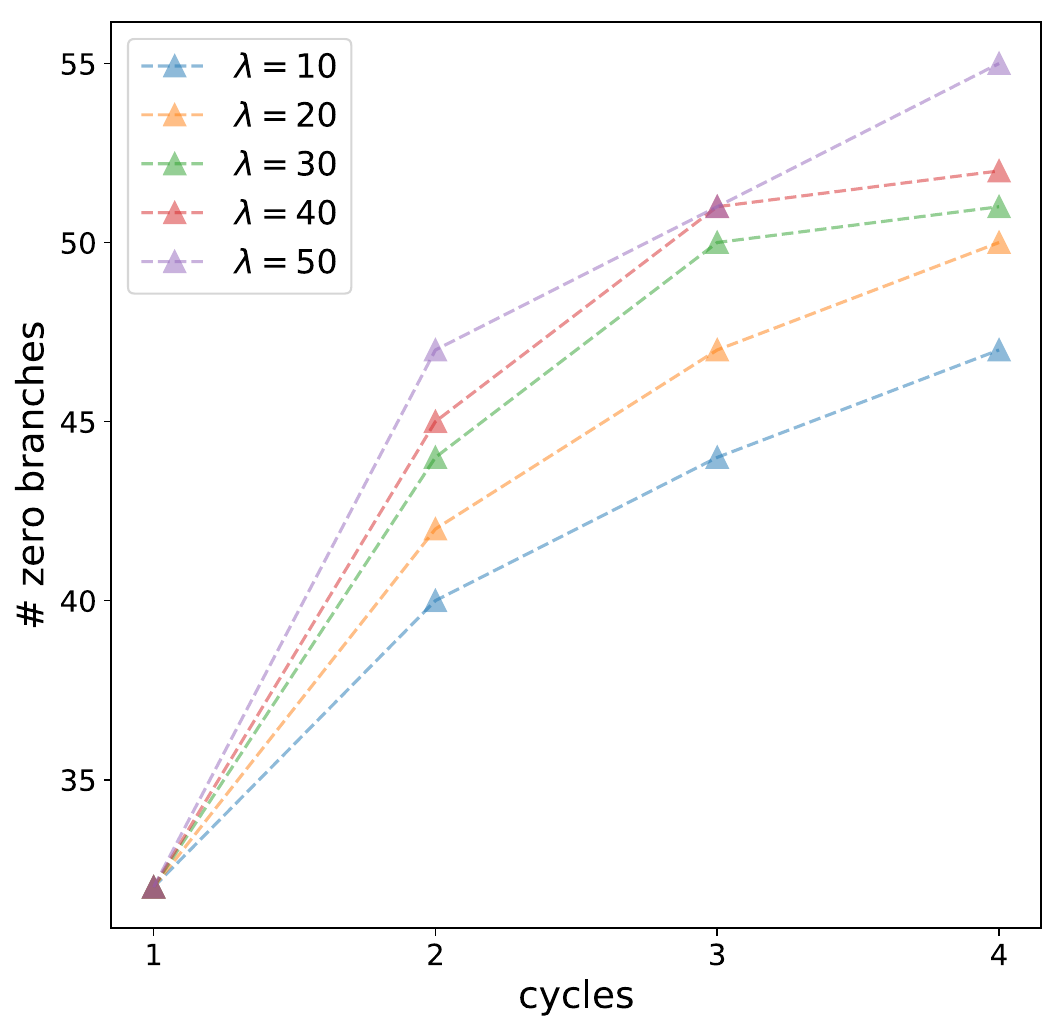}\hspace{20pt}
\includegraphics[width=0.45\textwidth]{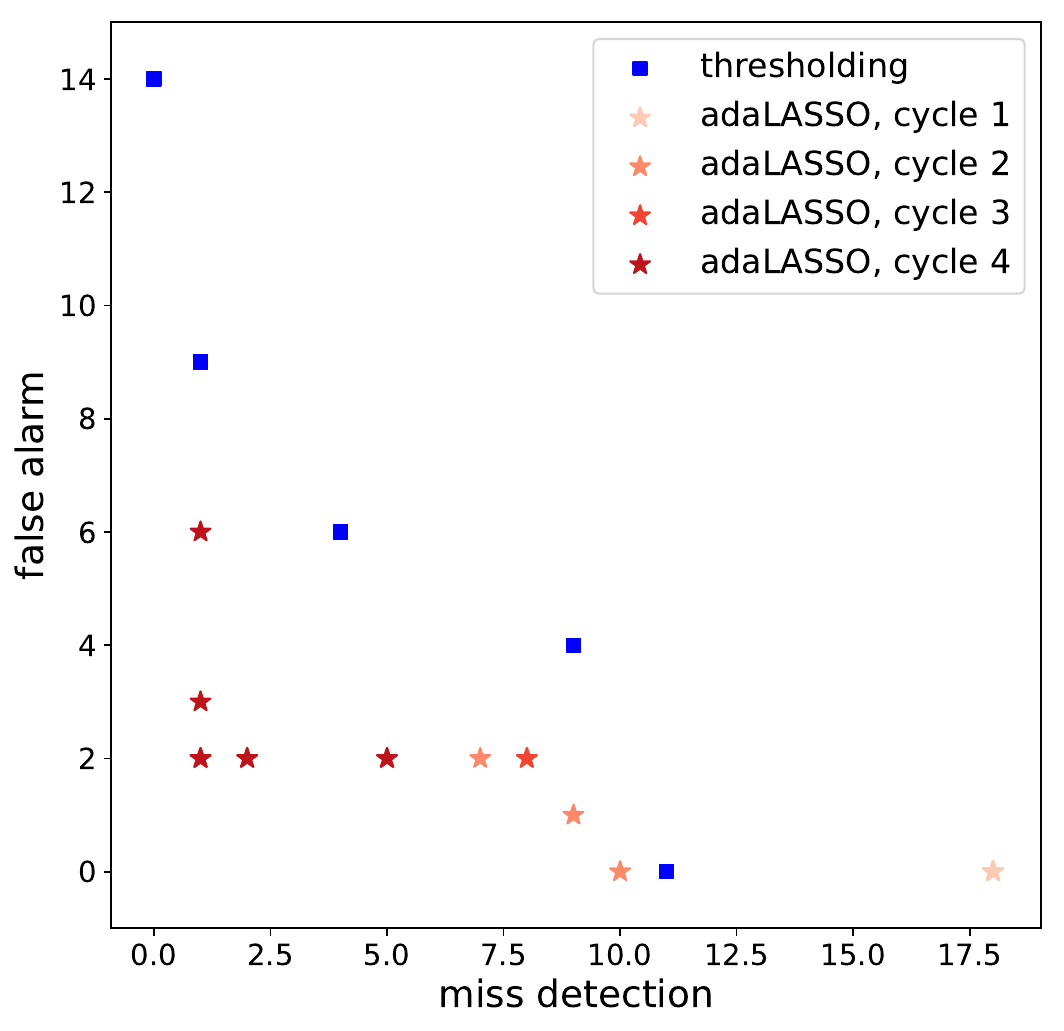}
\end{center}
\caption{Topological consistency comparison of different phylogenetic LASSO procedures on simulation 2. $\gamma=1.1$.}\label{fig:consistencygamma11}
\end{figure}

\begin{figure}[ht]
\begin{center}
\includegraphics[width=0.7\textwidth]{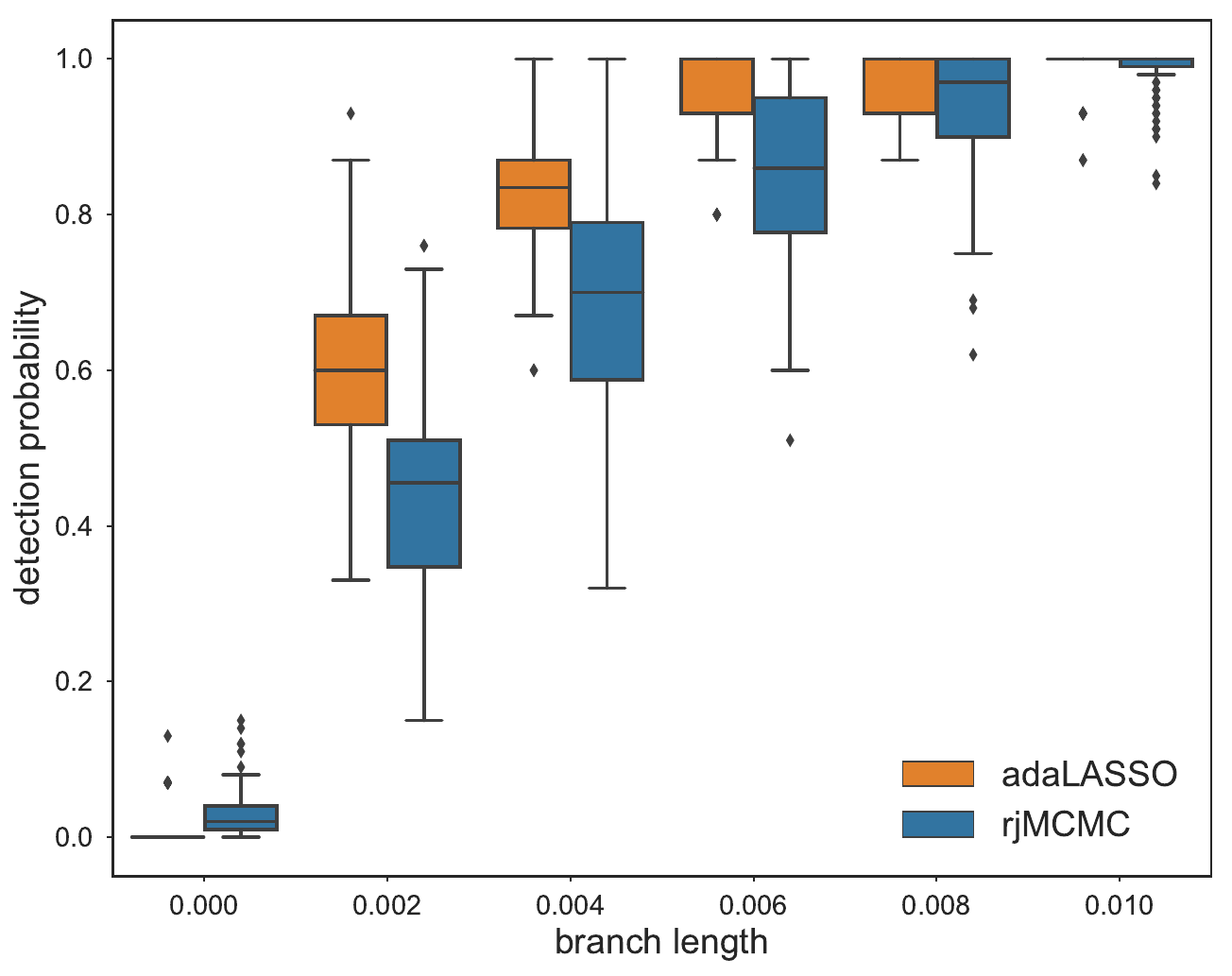}
\end{center}
\caption{Box plot showing performance of multistep adaptive phylogenetic LASSO and rjMCMC at detecting short branches. $\gamma=1.1$}\label{fig:detectiongamma11}
\end{figure}

\end{document}